%% file: partial.tex
\documentclass{llncs}

\usepackage{amsmath}
\usepackage{amssymb}
\usepackage[pdftitle={Extending partial representations of function graphs and permutation graphs}]{hyperref}
\usepackage{pstricks}
\usepackage[dvips]{graphicx}

\let\leq\leqslant
\let\geq\geqslant
\let\subset\subseteq
\let\subsetneq\varsubsetneq

\def\setR{\mathbb{R}}
\def\calC{\mathcal{C}}
\def\calF{\mathcal{F}}
\def\calG{\mathcal{G}}
\def\calM{\mathcal{M}}
\def\calS{\mathcal{S}}

\DeclareMathOperator{\Reg}{Reg}

\DeclareMathOperator{\dom}{dom}

\DeclareMathSymbol{\upset}{\mathclose}{symbols}{34}
\DeclareMathSymbol{\downset}{\mathclose}{symbols}{35}

\allowdisplaybreaks[4]

\begin{document}

\title{Extending partial representations of function graphs and permutation graphs}
\author{Pavel Klav\'{\i}k\inst{1}
\and Jan Kratochv\'{\i}l\inst{1}
\and Tomasz Krawczyk\inst{2}
\and Bartosz Walczak\inst{2}
}
\institute{Department of Applied Mathematics,\\Faculty of Mathematics and Physics, Charles University,\\\email{\{klavik,honza\}@kam.mff.cuni.cz}\and Department of Theoretical Computer Science,\\Faculty of Mathematics and Computer Science, Jagiellonian University,\\\email{\{krawczyk,walczak\}@tcs.uj.edu.pl}}

\maketitle

\thispagestyle{plain}

\begin{abstract}
Function graphs are graphs representable by intersections of continuous real-valued functions on the interval $[0,1]$ and are known to be exactly the complements of comparability graphs.
As such they are recognizable in polynomial time.
Function graphs generalize permutation graphs, which arise when all functions considered are linear.

We focus on the problem of extending partial representations, which generalizes the recognition problem.
We observe that for permutation graphs an easy extension of Golumbic's comparability graph recognition algorithm can be exploited.
This approach fails for function graphs.
Nevertheless, we present a polynomial-time algorithm for extending a partial representation of a graph by functions defined on the entire interval $[0,1]$ provided for some of the vertices.
On the other hand, we show that if a partial representation consists of functions defined on subintervals of $[0,1]$, then the problem of extending this representation to functions on the entire interval $[0,1]$ becomes NP-complete.
\end{abstract}

\input{introduction.tex}

\input{comp_and_perm.tex}

\section{Extending partial representation of a poset}\label{sec:poset}

Before we deal with function representations of graphs, we study representations of posets by continuous functions $[0,1]\to\setR$.
By a poset we mean a transitively oriented graph.
We write $u<_Pv$ to denote that there is an edge from $u$ to $v$ in a poset $P$.
Since we are interested in algorithmic problems, we have to choose some discrete description of the functions, and the particular choice does not matter as long as we can convert from and to other descriptions in polynomial time.
Here we restrict our attention to piecewise linear continuous functions.
Specifically, each function $f:[0,1]\to\setR$ that we consider is described by a tuple $(x_0,f(x_0)),\ldots,(x_k,f(x_k))$ of points in $[0,1]\times\setR$ with $0=x_0<\ldots<x_k=1$ so that $f$ is linear on every interval $[x_i,x_{i+1}]$.
We denote the family of such functions by $\calF$.
Note that every representation by continuous functions $[0,1]\to\setR$ can be approximated by an equivalent representation by functions from $\calF$.
We define a natural order $<$ on $\calF$ by setting $f<g$ whenever $f(x)<g(x)$ holds for every $x\in[0,1]$.
A \emph{representation} of a poset $P$ is a mapping $\phi:V(P)\to\calF$ such that
\[\forall u,v\in V(P)\colon(u<_Pv\iff\phi(u)<\phi(v)).\]
It is worth to note that every poset has a representation of this kind, see \cite{GRU}.
A \emph{partial representation} of a poset $P$ is a mapping $\phi:R\to\calF$ which is a representation of the subposet $P[R]$ induced on a set $R\subset V(P)$.

In this section we provide a polynomial-time algorithm solving the following problem: given a poset $P$ and its partial representation $\phi:R\to\calF$, decide whether $\phi$ is extendable to a representation of $P$.
Thus for the remainder of this section we assume that $P$ is a poset, $R\subset V(P)$, and $\phi:R\to\calF$ is a partial representation of $P$.

For a function $f\in\calF$ we define
\begin{align*}
f\upset&=\{(x,y)\in[0,1]\times\setR:y>f(x)\},\\
f\downset&=\{(x,y)\in[0,1]\times\setR:y<f(x)\}.
\end{align*}
For every vertex $u$ of $P$ we define a set $\Reg(u)\subset[0,1]\times\setR$, called the \emph{region} of $u$, as follows.
If $u\in R$, then $\Reg(u)=\phi(u)$.
Otherwise,
\[\Reg(u)=\bigcap\{\phi(a)\downset:a\in R\text{ and }a>_Pu\}\cap\bigcap\{\phi(a)\upset:a\in R\text{ and }a<_Pu\}.\]
It follows that the function representing $u$ in any representation of $P$ extending $\phi$ must be contained entirely within $\Reg(u)$.
See Fig.~\ref{fig:regions} for an illustration.

\begin{figure}[tb]
\centering\input{function_graphs.tex}
\caption{A poset $P$ and its partial representaion $\phi:\{a,b,c,d\}\to\calF$.
The diagram in the middle shows $\Reg(y)$ and a feasible $\psi(v)$ for a representaion $\psi$ of $P$ extending $\phi$.
The diagram to the right shows a representaion $\psi$ of $P$ extending $\phi$.}
\label{fig:regions}
\end{figure}

\begin{lemma}\label{lem:regions}
There is a representation of\/ $P$ extending\/ $\phi$ if and only if any two incomparable vertices\/ $u$ and\/ $v$ of\/ $P$ satisfy\/ $\Reg(u)\cap\Reg(v)\neq\emptyset$.
\end{lemma}

\noindent The proof is in the Appendix.

Lemma \ref{lem:regions} directly yields a polynomial-time algorithm for deciding whether $P$ has a representation extending $\phi$.
Indeed, the lower or upper boundary of $\Reg(u)$ (if exists) is a function from $\calF$ whose description can be easily computed from the descriptions of the functions $\phi(a)$ with $a\in R$ and $a<_Pu$ or $a>_Pu$, respectively.
Having the descriptions of the lower and upper boundaries of all regions, we can easily check whether the intersection of any two of them is empty.


\section{Modular decomposition}\label{sec:decomposition}

The main tool that we use for constructing a polynomial-time algorithm for extending partial representations of function graphs is modular decomposition, also known as substitution decomposition.
In this section we briefly discuss this concept and its connection to transitive orientations of graphs.

A graph is \emph{empty} if it has no edges.
A graph is \emph{complete} if it has all possible edges.
A set $M\subset V(G)$ is a \emph{module} of $G$ if every vertex in $V(G)-M$ is adjacent to either all or none of the vertices in $M$.
The singleton sets and the whole $V(G)$ are the \emph{trivial} modules of $G$.
A non-empty graph is \emph{prime} if it has no modules other than the trivial ones.
A module $M$ is \emph{strong} if every module $N$ satisfies $N\subset M$, $M\subset N$, or $M\cap N=\emptyset$.
We denote the family of non-singleton strong modules of $G$ by $\calM(G)$.
A strong module $M\subsetneq V(G)$ is \emph{maximal} if there is no strong module $N$ with $M\subsetneq N\subsetneq V(G)$.
When $G$ is a graph with at least two vertices, the maximal strong modules of $G$ form a partition of $V(G)$, which we denote by $\calC(G)$.
It is easy to see that a set $M\subsetneq V(G)$ is a strong module of $G$ if and only if $M$ is a strong module of $G[N]$ for some $N\in\calC(G)$.
Applying this observation recursively, we see that the strong modules of $G$ form a rooted tree, called the \emph{modular decomposition} of $G$, in which
\begin{itemize}
\item $V(G)$ is the root;
\item $\calC(G[M])$ are the children of every $M\in\calM(G)$;
\item the singleton modules are the leaves.
\end{itemize}
In particular, $G$ has at most $2|V(G)|-1$ strong modules in total.

Any two distinct strong modules $M,N\subsetneq V(G)$ can be either \emph{adjacent}, which means that any two vertices $u\in M$ and $v\in N$ are adjacent in $G$, or \emph{non-adjacent}, which means that no two vertices $u\in M$ and $v\in N$ are adjacent in $G$.
When $M$ and $N$ are two adjacent strong modules of $G$ and $P$ is a transitive orientation of $G$, we write $M<_PN$ to denote that $u<_Pv$ for all $u\in M$ and $v\in N$.

\begin{theorem}[Gallai \cite{Gallai}]\label{thm:orientation}
Let\/ $M$ and\/ $N$ be two adjacent strong modules of\/ $G$.
Every transitive orientation\/ $P$ of\/ $G$ satisfies either\/ $M<_PN$ or\/ $M>_PN$.
\end{theorem}

For a module $M\in\calM(G)$, we call the adjacency graph of $\calC(G[M])$ the \emph{quotient} of $M$ and denote it by $G[M]/\calC(G[M])$, and we call a transitive orientation of $G[M]/\calC(G[M])$ simply a \emph{transitive orientation} of $M$.

\begin{theorem}[Gallai \cite{Gallai}]\label{thm:correspondence}
The transitive orientations of\/ $G$ and the tuples of transitive orientations of non-singleton strong modules of\/ $G$ are in a one-to-one correspondence\/ $P\leftrightarrow(P_M)_{M\in\calM(G)}$ given by\/ $M_1<_{P_M}M_2\iff M_1<_PM_2$ for any\/ $M\in\calM(G)$ and\/ $M_1,M_2\in\calC(M)$.
In particular, $G$ is a comparability graph if and only if\/ $G[M]/\calC(G[M])$ is a comparability graph for every\/ $M\in\calM(G)$.
\end{theorem}

\begin{theorem}[Gallai \cite{Gallai}]\label{thm:decomposition}
Let\/ $M$ be a non-singleton strong module of\/ $G$.
\begin{enumerate}
\item If\/ $G[M]$ is not connected, then the maximal strong modules of\/ $G[M]$ are the connected components of\/ $G[M]$ and\/ $G[M]/\calC(G[M])$ is an empty graph.
\item If\/ $\overline{G[M]}$ is not connected, then the maximal strong modules of\/ $G[M]$ are the connected components of\/ $\overline{G[M]}$ and\/ $G[M]/\calC(G[M])$ is a complete graph.
\item If\/ $G[M]$ and\/ $\overline{G[M]}$ are connected, then\/ $G[M]/\calC(G[M])$ is a prime graph.
\end{enumerate}
\end{theorem}

Theorem \ref{thm:decomposition} allows us to classify non-singleton strong modules into three types.
Namely, a non-singleton strong module $M$ of $G$ is
\begin{itemize}
\item\emph{parallel} when $G[M]/\calC(G[M])$ is empty;
\item\emph{serial} when $G[M]/\calC(G[M])$ is complete;
\item\emph{prime} when $G[M]/\calC(G[M])$ is prime.
\end{itemize}
Every parallel module has just one transitive orientation---there is nothing to orient in an empty quotient.
Every serial module with $k$ children has exactly $k!$ transitive orientations corresponding to the $k!$ permutations of the children.
Finally, for prime modules we have the following.

\begin{theorem}[Gallai \cite{Gallai}]\label{thm:prime}
Every prime module of a comparability graph has exactly two transitive orientations, one being the reverse of the other.
\end{theorem}

Golumbic \cite{Golumbic1,Golumbic2} showed that the problems of computing the modular decomposition of a graph, computing the two transitive orientations of a prime comparability graph, and deciding whether a graph is a comparability graph are polynomial-time solvable.
Actually, the first two of these problems can be solved in linear time \cite{MS}.

\section{Extending partial represenation of a function graph}

In this section we provide a polynomial-time algorithm for extending partial representation of function graphs.
However, for convenience, instead of function graphs we deal with their complements---comparability graphs.
A \emph{representation} of a comparability graph $G$ is a representation of a transitive orientation of $G$, defined as in Section \ref{sec:poset}.
A \emph{partial representation} of $G$ is a representation of an induced subgraph of $G$.

Specifically, we prove that the following problem is polynomial-time solvable: given a comparability graph $G$ and its partial representation $\phi:R\to\calF$, decide whether $\phi$ is extendable to a representation of $G$.
Thus for the remainder of this section we assume that $G$ is a comparability graph, $R\subset V(G)$, and $\phi:R\to\calF$ is a partial representation of $G$.

A transitive orientation $P$ of $G$ \emph{respects} $\phi$ if $\phi$ is a partial representation of~$P$.
The idea of the algorithm is to look for a transitive orientation $P$ of $G$ that respects $\phi$ and satisfies $\Reg_P(u)\cap\Reg_P(v)\neq\emptyset$ for any two adjacent vertices $u$ and $v$ of $G$, where by $\Reg_P(u)$ we denote the region of $u$ with respect to $P$.
By Lemma \ref{lem:regions}, such a transitive orientation exists if and only if $\phi$ is extendable.
We make use of the modular decomposition of $G$ and Theorem \ref{thm:correspondence} to identify all transitive orientations of $G$.
We apply to $G$ a series of reductions, which ensure that every non-singleton strong module of $G$ has exactly one or two transitive orientations respecting $\phi$, while not changing the answer.
Finally, after doing all these reductions, we express the existence of a requested transitive orientation of $G$ by an instance of 2-SAT\@.

A strong module $M$ of $G$ is \emph{represented} if $M\cap R\neq\emptyset$.
Any vertex from $M\cap R$ is a \emph{representant} of $M$.
Clearly, if $M$ is represented, then all ancestors of $M$ in the modular decomposition of $G$ are represented as well.

The first step of the algorithm is to compute the modular decomposition of $G$, which can be done in polynomial time as commented at the end of the previous section.
Then, we apply three kinds of reductions, which modify $G$ and its modular decomposition but do not affect $R$ and $\phi$:
\begin{enumerate}
\item If there is a non-singleton non-represented module $M$ of $G$, then choose any vertex $u\in M$, remove $M-\{u\}$ (with all incident edges) from $G$ and from the nodes of the modular decomposition, and replace the subtree rooted at $M$ by the singleton module $\{u\}$ in the modular decomposition.
\item If there is a serial module $M$ of $G$ with two or more non-represented children, then we choose any non-represented child $N$ of $M$ and remove from $G$ and from the modular decomposition all other non-represented children of~$M$.
\item If there is a serial module $M$ of $G$ with two or more represented children and some non-represented children, then we remove from $G$ and from the modular decomposition all non-represented children of $M$.
\end{enumerate}

\begin{lemma}\label{lem:reductions}
The graph\/ $G$ has a representation extending\/ $\phi$ if and only if the graph\/ $G'$ obtained from\/ $G$ by reductions 1--3 has a representation extending\/ $\phi$.
\end{lemma}

\noindent The proof is in the Appendix.

We apply reductions 1--3 in any order until none of them is applicable any more, that is, we are left with a graph $G$ such that
\begin{itemize}
\item every non-singleton strong module of $G$ is represented,
\item every serial module of $G$ has at most one non-represented child,
\item every serial module of $G$ with at least two represented children has no non-represented child.
\end{itemize}
For such $G$ we have the following.

\begin{lemma}\label{lem:no-of-orientations}
Let\/ $M$ be a non-singleton strong module of\/ $G$.
If\/ $M$ is
\begin{itemize}
\item a serial module with a non-represented child,
\item a prime module with no two adjacent represented children,
\end{itemize}
then\/ $M$ has exactly two transitive orientations, one being the reverse of the other, both respecting\/ $\phi$.
Otherwise, $M$ has just one transitive orientation respecting\/~$\phi$.
\end{lemma}

\noindent The proof is in the Appendix.

\begin{lemma}\label{lem:reg-dependence}
Let\/ $u$ be a non-represented vertex of\/ $G$ and\/ $M$ be the parent of\/ $\{u\}$ in the modular decomposition of\/ $G$.
For transitive orientations\/ $P$ of\/ $G$ respecting\/ $\phi$, the set\/ $\Reg_P(u)$ is determined only by the transitive orientation of\/ $M$ induced by\/ $P$.
\end{lemma}

\begin{proof}
Let $a$ be a represented vertex of $G$ adjacent to $u$.
We show that the orientation of the edge $au$ either is the same for all transitive orientations of $G$ respecting $\phi$ or depends only on the transitive orientation of $M$.
This suffices for the conclusion of the lemma, as the set $\Reg_P(u)$ is determined by the orientations of edges connecting $u$ with represented vertices of $G$.

If $a\in M$, then clearly the orientation of the edge $au$ depends only on the transitive orientation of $M$.
Thus suppose $a\notin M$.
Let $b$ be a representant of $M$.
Since $M$ is a module, $a$ is adjacent to $b$ as well.
By Theorem \ref{thm:orientation}, the orientations of the edges $au$ and $ab$ are the same for every transitive orientation of $G$.
The orientation of $ab$ and thus of $au$ in any transitive orientation of $G$ respecting $\phi$ is fixed by $\phi$.
Therefore, all transitive orientations of $G$ respecting $\phi$ yield the same orientation of the edge $au$.\qed
\end{proof}

By Lemmas \ref{lem:no-of-orientations} and \ref{lem:reg-dependence}, for every $u\in V(G)$, all transitive orientations $P$ of $G$ respecting $\phi$ yield at most two different regions $\Reg_P(u)$.
We can compute them following the argument in the proof of Lemma \ref{lem:reg-dependence}.
Namely, if $u\in R$ then $\Reg_P(u)=\phi(u)$, otherwise we find the neighbors of $u$ in $R$ that bound $\Reg_P(u)$ from above and from below, depending on the orientation of $M$, and compute the geometric representation of $\Reg_P(u)$ as for the poset problem in Section~\ref{sec:poset}.

Now, we describe a reduction of the problem to 2-SAT\@.
For every $M\in\calM(G)$ with two transitive orientations respecting $\phi$, we introduce a boolean variable $x_M$. The two valuations of $x_M$ represent the two transitive orientations of $M$.
We write a formula of the form $\alpha=\alpha_1\wedge\ldots\wedge\alpha_m$, where each clause $\alpha_j$ is a literal or an alternative of two literals of the form $x_M$ or $\neg x_M$, as follows.
By Lemma \ref{lem:reg-dependence}, the set $\Reg_P(u)$ for any vertex $u$ is either the same for all valuations or determined by the valuation of just one variable.
Therefore, for any two non-adjacent vertices $u$ and $v$, whether $\Reg_P(u)\cap\Reg_P(v)\neq\emptyset$ depends on the valuation of at most two variables.
For every valuation that yields $\Reg_P(u)\cap\Reg_P(v)=\emptyset$ we write a clause forbidding this valuation.
Clearly, the resulting formula $\alpha$ is satisfiable if and only if $G$ has a transitive orientation $P$ respecting $\phi$ and such that $\Reg_P(u)\cap\Reg_P(v)\neq\emptyset$ for any non-adjacent $u,v\in V(G)$, which by Lemma \ref{lem:regions} holds if and only if $\phi$ is extendable to a representation of $G$.
We can test whether $\alpha$ is satisfiable in polynomial time by a classic result of Krom~\cite{Krom}
(see also \cite{APT} for a linear-time algorithm).


\section{Extending partial represenation of a function graph by partial functions}\label{sec:partial-functions}

Let $\calF^\circ$ denote the family of piecewise linear continuous functions $I\to\setR$ with $I$ being a closed subinterval of $[0,1]$.
We describe such a function $f$ by a tuple $(x_0,f(x_0)),\ldots,(x_k,f(x_k))$ of points in $I\times\setR$ with $x_0<\ldots<x_k$ and $[x_0,x_k]=I$ so that $f$ is linear on every interval $[x_i,x_{i+1}]$.
We denote the interval $I$ that is the domain of $f$ by $\dom f$.
For convenience, we also put the empty function (with empty domain) to $\calF^\circ$.
We say that a mapping $\psi:U\to\calF$ \emph{extends} a mapping $\phi:U\to\calF^\circ$ if we have $\psi(u)|_{\dom\phi(u)}=\phi(u)$ for every $u\in U$.

We define the notions of a partial representation of a poset or graph by partial functions, which generalize partial representations by functions defined on the entire interval $[0,1]$ and discussed earlier in the paper.
A mapping $\phi:V(P)\to\calF^\circ$ is a \emph{partial representation} of a poset $P$ if the following is satisfied for any $u,v\in V(P)$: if $u<_Pv$, then $\phi(u)(x)<\phi(v)(x)$ for every $x\in\dom\phi(u)\cap\dom\phi(v)$.
A mapping $\phi:V(G)\to\calF^\circ$ is a \emph{partial representation} of a comparability graph $G$ if $\phi$ is a partial representation of some transitive orientation of $G$.
The domain of $\phi$ is the whole set of vertices, as we may map non-represented vertices to the empty function.

In this section we prove that the following problem is NP-complete: given a comparability graph $G$ and a partial representation $\phi:V(G)\to\calF^\circ$ of $G$, decide whether $\phi$ is extendable to a representation $\psi:V(G)\to\calF$ of~$G$.

For a function $f\in\calF^\circ$ we define
\begin{align*}
f^\star&=f\cup(([0,1]-\dom f)\times\setR),\\
f\upset&=\{(x,y)\in\dom f\times\setR:y>f(x)\}\cup(([0,1]-\dom f)\times\setR),\\
f\downset&=\{(x,y)\in\dom f\times\setR:y<f(x)\}\cup(([0,1]-\dom f)\times\setR).
\end{align*}
Let $P$ be a poset and $\phi:V(P)\to\calF^\circ$ be a partial representation of $P$.
For every vertex $u$ of $P$ we define a set $\Reg(u)\subset[0,1]\times\setR$, called the \emph{region} of $u$, by
\[\Reg(u)=\phi(u)^\star\cap\bigcap\{\phi(a)\downset:a>_Pu\}\cap\bigcap\{\phi(a)\upset:a<_Pu\}.\]
It follows that the function representing $u$ in any representation of $P$ extending $\phi$ must be contained entirely within $\Reg(u)$.
Lemma \ref{lem:regions} generalizes verbatim to representations by partial functions (proof in the Appendix).

\begin{lemma}\label{lem:regions-partial}
There is a representation of\/ $P$ extending\/ $\phi$ if and only if any two incomparable vertices\/ $u$ and\/ $v$ of\/ $P$ satisfy\/ $\Reg(u)\cap\Reg(v)\neq\emptyset$.
\end{lemma}

Lemma \ref{lem:regions-partial} shows that the problem of deciding whether a partial representation of a poset by partial functions is extendable is in P: a polynomial-time algorithm just tests whether $\Reg(u)\cap\Reg(v)\neq\emptyset$ for any two incomparable vertices $u$ and $v$ of the poset.
It follows that the problem of deciding whether a partial representation of a comparability graph by partial functions is extendable is in NP: a non-deterministic polynomial-time algorithm can guess a transitive orientation and solve the resulting poset problem.

\begin{figure}[tb]
\centering\input{graph.tex}
\caption{The graph $G$ and its modular decomposition. The edges whose orientation is fixed by $\phi$ are drawn directed.}
\label{fig:graph}
\end{figure}

\begin{figure}[tb]
\centering\input{clause.tex}
\caption{The partial functions representing the vertices from $X_1$ and $A_1$, for a clause $\alpha_1=\alpha_1^1\vee\alpha_1^2\vee\alpha_1^3$ with $\alpha_1^1=x_1$, $\alpha_1^2=x_2$, and $\alpha_1^3=\neg x_3$.}
\label{fig:clause}
\end{figure}

\begin{figure}[tb]
\centering\input{region_clause.tex}
\caption{$\Reg_P(\alpha_{1,1}^1)$ for a clause $\alpha_1=\alpha_1^1\vee\alpha_1^2\vee\alpha_1^3$ with $\alpha_1^1=x_1$, $\alpha_1^2=x_2$, and $\alpha_1^3=\neg x_3$ and an orientation of the module $A_1$ such that $A_1^1>_PA_1^2,A_1^3$.}
\label{fig:region-clause}
\end{figure}

\begin{figure}[tb]
\centering\input{variable.tex}
\caption{$\Reg_P(x_{1,1}^-)$ for $X_1^-<_PX_1^+$ (to the left) and $X_1^+<_PX_1^-$ (to the right).}
\label{fig:region-variable}
\end{figure}

To prove that the latter problem is NP-hard, we show a polynomial-time reduction from 3-SAT\@.
Let $\alpha=\alpha_1\wedge\ldots\wedge\alpha_m$ be a boolean formula over variables $x_1,\ldots,x_n$, where each clause $\alpha_j$ is of the form $\alpha_j=\alpha_j^1\vee\alpha_j^2\vee\alpha_j^3$ with $\alpha_j^k\in\{x_1,\neg x_1,\ldots,x_n,\neg x_n\}$.
We construct a comparability graph $G$ and its partial representation $\phi:V(G)\to\calF^\circ$ that is extendable if and only if $\alpha$ is satisfiable.
The vertex set of $G$ consists of groups $X_i$ of six vertices corresponding to variables, groups $A_j$ of thirteen vertices corresponding to clauses, and two special vertices $p$ and $q$.
The edges and the modular decomposition of $G$ are illustrated in Fig.\ \ref{fig:graph}.
The partial representation $\phi$ is illustrated in Fig.\ \ref{fig:clause}, and the orientations of edges that are common for all transitive orientations of $G$ respecting $\phi$ are shown again in Fig.\ \ref{fig:graph}.
Every valuation satisfying $\alpha$ corresponds to a transitive orientation $P$ of $G$ that respects $\phi$ and satisfies $\Reg_P(u)\cap\Reg_P(v)\neq\emptyset$ for any two incomparable vertices $u$ and $v$, as follows: $x_i$ is true if and only if $X_i^+<_PX_i^-$, and $\alpha_j^k$ is satisfied if $A_j^k$ is maximal among $A_j^1,A_j^2,A_j^3$ with respect to $<_P$.
This together with Lemma \ref{lem:regions-partial} implies that $\alpha$ is satisfiable if and only if $\phi$ is extendable.
See the Appendix for the full proof.
Figures \ref{fig:region-clause} and \ref{fig:region-variable} illustrate how the regions of the vertices depend on the chosen orientation.


\newpage

\appendix
\section*{Appendix: Omitted proofs}

\input{comp_and_perm_appendix.tex}

\section{Proofs of Lemmas \ref{lem:regions}, \ref{lem:reductions}, \ref{lem:no-of-orientations}, and \ref{lem:regions-partial}}

For a poset $P$ and a set $S\subset V(P)$, a function $\sigma:S\to\setR$ is \emph{$P$-compatible} if we have $\sigma(u)<\sigma(v)$ for any $u,v\in S$ with $u<_Pv$.

\begin{lemma}\label{lem:functions}
For every set\/ $S\subset V(P)$ and every\/ $P$-compatible function\/ $\sigma:S\to\setR$ there is a\/ $P$-compatible function\/ $\tau:V(P)\to\setR$ with\/ $\tau|_S=\sigma$.
\end{lemma}

\begin{proof}
The proof goes by induction on $|V(P)-S|$.
For $S=V(P)$ the conclusion holds trivially.
Now, assume $S\subsetneq V(P)$.
Choose any $u\in V(P)-S$.
Define
\begin{align*}
y_1&=\max\{\sigma(v):v\in S\text{ and }v<_Pu\},\\
y_2&=\min\{\sigma(v):v\in S\text{ and }v>_Pu\},
\end{align*}
where we adopt the convention that $\max\emptyset=-\infty$ and $\min\emptyset=+\infty$.
If $y_1>-\infty$ and $y_2<+\infty$, then there are $v_1,v_2\in S$ with $\sigma(v_1)=y_1$, $\sigma(v_2)=y_2$, and $v_1<_Pu<_Pv_2$, which by $P$-compatibility of $\sigma$ implies $y_1<y_2$.
Choose any $y\in(y_1,y_2)$.
Set $S'=S\cup\{u\}$, $\sigma'(u)=y$, and $\sigma'(v)=\sigma(v)$ for every $v\in S$.
Clearly, the function $\sigma':S'\to\setR$ is $P$-compatible.
Therefore, by the induction hypothesis, there is a $P$-compatible function $\tau:V(P)\to\setR$ with $\tau|_{S'}=\sigma'$.
It follows that $\tau|_S=\sigma$, which completes the proof.\qed
\end{proof}

\begin{proof}[Lemma \ref{lem:regions}]
Suppose that $P$ has a representation $\psi:V(G)\to\calF$ extending $\phi$.
Choose any two incomparable vertices $u$ and $v$ of $P$.
The functions $\psi(u)$ and $\psi(v)$ must intersect at some point, which belongs to both $\Reg(u)$ and $\Reg(v)$.
Hence $\Reg(u)\cap\Reg(v)\neq\emptyset$.

Now, suppose that we have $\Reg(u)\cap\Reg(v)\neq\emptyset$ for any two incomparable vertices $u$ and $v$ of $P$.
For convenience, we assume that the coordinates $x_0,\ldots,x_k$ used to describe functions $\phi(a)$ are common for all $a\in R$ (we can achieve this by adding new coordinates to the descriptions where necessary).
We show how to construct a representation $\psi:V(G)\to\calF$ of $P$ extending $\phi$.

First, we define the functions $\psi(*)$ on the coordinates $x_0,\ldots,x_k$.
Let $i\in\{0,\ldots,n\}$.
Define a function $\rho:R\to\setR$ by setting $\rho(a)=\phi(a)(x_i)$ for every $a\in R$.
Clearly, $\rho$ is $P$-compatible.
By Lemma \ref{lem:functions}, there is a $P$-compatible function $\tau:V(P)\to\setR$ with $\tau|_R=\rho$.
Set $\psi(u)(x_i)=\tau(u)$ for every $u\in V(P)$.

Now, for any two incomparable vertices $u$ and $v$ of $P$, we define the functions $\psi(*)$ on an additional coordinate different from $x_0,\ldots,x_k$ and from the coordinates chosen for the other pairs $u,v$, to ensure that $\psi(u)$ and $\psi(v)$ intersect.
The set $X=\{x\in[0,1]:\Reg(u)\cap\Reg(v)\cap(\{x\}\times\setR)\neq\emptyset\}$ is non-empty.
If $X$ has an isolated point $x$, then we have $u,v\in R$ and $\phi(u)(x)=\phi(v)(x)$.
Otherwise, $X$ is an infinite subset of $[0,1]$.
Choose $x\in[0,1]$ and $y\in\setR$ so that $(x,y)\in\Reg(u)\cap\Reg(v)$ and $x$ is different from $x_0,\ldots,x_k$ and every other $x$ already chosen for a different pair $u,v$.
Let $S=R\cup\{u,v\}$.
Define a function $\sigma:S\to\setR$ by setting $\sigma(u)=\sigma(v)=y$ and $\sigma(a)=\phi(a)(x)$ for every $a\in R$.
It follows from the definition of $\Reg(u)$ and $\Reg(v)$ that $\sigma$ is $P$-compatible.
By Lemma \ref{lem:functions}, there is a $P$-compatible function $\tau:V(P)\to\setR$ with $\tau|_S=\sigma$.
Set $\psi(w)(x)=\tau(w)$ for every $w\in V(P)$.
This ensures that $\psi(u)$ and $\psi(v)$ intersect at $(x,y)$.

We have defined the functions $\psi(*)$ on a finite set of coordinates $X$.
Now, we extend them to the entire $[0,1]$ so that each of them is linear between any two consecutive coordinates from $X$.
It is clear from the construction that $\psi(a)=\phi(a)$ for every $a\in R$.
Since all functions $\tau:V(P)\to\setR$ used for defining $\psi(*)$ on the coordinates from $X$ are $P$-compatible, $u<_Pv$ implies $\psi(u)<\psi(v)$.
Finally, we have guaranteed that $\psi(u)$ and $\psi(v)$ intersect for any two incomparable vertices $u$ and $v$ of $P$.
Thus $\psi$ is a representation of $P$ extending~$\phi$.\qed
\end{proof}

\begin{proof}[Lemma \ref{lem:reductions}]
It suffices that we consider only single steps of reductions 1--3.
Clearly, if $G$ has a representation extending $\phi$, then $G'$ has a representation extending $\phi$, as all reductions only remove vertices from $G$.
To prove the converse implication, suppose that $G'$ has a representation extending $\phi$.
By Lemma \ref{lem:regions}, this implies that $G'$ has a transitive orientation $P'$ that respects $\phi$ and satisfies $\Reg_{P'}(u)\cap\Reg_{P'}(v)\neq\emptyset$ for every two incomparable vertices $u$ and $v$.
We show that $G$ has a transitive orientation $P$ with the same property, which again by Lemma \ref{lem:regions} yields the conclusion.

Suppose that $G'$ is obtained from $G$ by a single step of reduction~1.
Let $M$ and $u$ be the module and the vertex chosen for the reduction step.
Choose a transitive orientation $P$ of $G$ that induces $P'$ on $V(G')$ and arbitrarily orients the internal nodes removed from the modular decomposition by the reduction.
By Theorem \ref{thm:orientation}, for every $v\in V(G)-V(G')$ we have
\begin{align*}
\{a\in R:a<_Pv\}&=\{a\in R:a<_Pu\},\\
\{a\in R:a>_Pv\}&=\{a\in R:a>_Pu\},
\end{align*}
and therefore $\Reg_P(v)=\Reg_P(u)=\Reg_{P'}(u)$.
This implies that $\Reg_P(u)\cap\Reg_P(v)\neq\emptyset$ is satisfied for any two non-adjacent $u,v\in V(G)$.

Now, suppose that $G'$ is obtained from $G$ by a single step of reduction~2.
Let $M$ and $N$ be the modules chosen for the reduction step.
Choose a transitive orientation $P$ of $G$ that induces $P'$ on $V(G')$, orients $M$ so that $N'<_PM'\iff N<_{P'}M'$ for every represented child $M'$ and every non-represented child $N'$ of $M$ in the modular decomposition of $G$, and arbitrarily orients the internal nodes removed from the modular decomposition by the reduction.
Again, by Theorem \ref{thm:orientation}, for every $v\in V(G)-V(G')$ we have
\begin{align*}
\{a\in R:a<_Pv\}&=\{a\in R:a<_Pu\},\\
\{a\in R:a>_Pv\}&=\{a\in R:a>_Pu\},
\end{align*}
which yields the same conclusion as for reduction~1.

Finally, suppose that $G'$ is obtained from $G$ by a single step of reduction~3.
Let $M$ the module chosen for the reduction step.
Let $M_1$ and $M_2$ be any two represented children of $M$ in the modular decomposition of $G$ that are consecutive in the order $<_{P'}$ on the represented children of $M$.
Choose a transitive orientation $P$ of $G$ that induces $P'$ on $V(G')$, orients $M$ so that $M_1<_PN<_PM_2$ for every non-represented child $N$ of $M$ in the modular decomposition of $G$, and arbitrarily orients the internal nodes removed from the modular decomposition by the reduction.
Let $a_1\in M_1\cap R$, $a_2\in M_2\cap R$, and $v\in V(G)-V(G')$.
By Theorem \ref{thm:orientation}, we have
\begin{align*}
\{a\in R:a<_Pv\}&=\{a_1\}\cup\{a\in R:a<_Pa_1\},\\
\{a\in R:a>_Pv\}&=\{a_2\}\cup\{a\in R:a>_Pa_2\}.
\end{align*}
This implies that $\Reg_P(v)$ is the region contained between $\phi(a_1)$ and $\phi(a_2)$.
In particular, for every other $w\in V(G)-V(G')$ we have $\Reg_P(v)=\Reg_P(w)$.
For every $w\in V(G')$ that is incomparable to $u$ in $P$ we have $\Reg_P(a_1)\cap\Reg_P(w)=\Reg_{P'}(a_1)\cap\Reg_{P'}(w)\neq\emptyset$ and $\Reg_P(a_2)\cap\Reg_P(w)=\Reg_{P'}(a_2)\cap\Reg_{P'}(w)\neq\emptyset$.
Thus we also have $\Reg_P(u)\cap\Reg_P(w)\neq\emptyset$.\qed
\end{proof}

\begin{proof}[Lemma \ref{lem:no-of-orientations}]
Let $M$ be a non-singleton strong module of $G$.
If $M$ is parallel, then it has just one transitive orientation at all.
If $M$ is serial and has no non-represented child, then it has only one transitive orientation respecting $\phi$, as $\phi$ determines the order of representants of the children of $M$.
If $M$ is serial and has one represented child and one non-represented child, then it has just two transitive orientations, one being the reverse of the other, both respecting $\phi$.
Finally, if $M$ is prime, then by Theorem \ref{thm:prime} it has exactly two transitive orientations, one being the reverse of the other.
In the latter case, if $M$ has two adjacent represented children, then $\phi$ determines the order of their representants, which shows that only one of the two transitive orientations of $M$ respects $\phi$.
Otherwise, both transitive orientations of $M$ respect $\phi$.\qed
\end{proof}

\begin{proof}[Lemma \ref{lem:regions-partial}]
The proof is exactly the same as the proof of Lemma \ref{lem:regions}, assuming that by $R$ occurring in the context of a coordinate $x\in[0,1]$ we mean the set $\{a\in V(P):x\in\dom\phi(a)\}$.\qed
\end{proof}

\section{Detailed proof of Theorem \ref{thm:partial_fun_ext}}

Below, we describe the construction of the graph $G$ in the reduction from 3-SAT in Section~\ref{sec:partial-functions}.
For simplicity, we assume that $m,n\geq 2$ and the literals $\alpha_j^1,\alpha_j^2,\alpha_j^3$ within each clause $\alpha_j$ are distinct.

The vertex set of $G$ consists of a set $X=X_1\cup\ldots\cup X_n$ of vertices corresponding to variables,
a set $A=A_1\cup\ldots\cup A_m$ of vertices corresponding to clauses, and two special vertices $p$ and $q$, see Fig.\ \ref{fig:graph}.
Each set $X_i$ corresponding to the variable $x_i$ contains six vertices partitioned into two groups:
\[X_i^+=\{x_{i,1}^+,x_{i,2}^+,x_{i,3}^+\},\quad X_i^-=\{x_{i,1}^-,x_{i,2}^-,x_{i,3}^-\}.\]
Each set $A_j$ corresponding to the clause $\alpha_j$ contains thirteen vertices partitioned into four groups:
\begin{alignat*}{2}
A_j^k&=\{\alpha_{j,1}^k,\alpha_{j,2}^k,\alpha_{j,3}^k\}&\quad&\text{for }k\in\{1,2,3\},\\
T_j&=\{t_{j,1},t_{j,2},t_{j,3},t_{j,4}\}.
\end{alignat*}
The edge set of $G$ is defined so that
\begin{itemize}
\item the vertices from $X_i^+$ are adjacent to the vertices from $X_i^-$,
\item the vertices from $A_j^k$ are adjacent to the vertices from $A_j^\ell$ for $k\neq\ell$,
\item the vertices from $A_j^k$ are adjacent to the vertices from $T_j$,
\item $p$ is adjacent to the entire $A$,
\item $q$ is adjacent to the entire $X$,
\item $p$ is adjacent to $q$,
\item there are no other adjacencies.
\end{itemize}
The modular decomposition of $G$ looks as follows:
\begin{itemize}
\item $V(G)$ is a prime module with children $X$, $A$, $\{p\}$, and $\{q\}$;
\item $X$ is a parallel module with children $X_1,\ldots,X_n$;
\item each $X_i$ is a serial module with children $X_i^+$ and $X_i^-$;
\item each $X_i^+$ is a parallel module with children $\{x_{i,1}^+\}$, $\{x_{i,2}^+\}$, and $\{x_{i,3}^+\}$;
\item each $X_i^-$ is a parallel module with children $\{x_{i,1}^-\}$, $\{x_{i,2}^-\}$, and $\{x_{i,3}^-\}$;
\item $A$ is a parallel module with children $A_1,\ldots,A_m$;
\item each $A_j$ is a serial module with children $A_j^1$, $A_j^2$, $A_j^3$, and $T_j$;
\item each $A_j^k$ is a parallel module with children $\{\alpha_{j,1}^k\}$, $\{\alpha_{j,2}^k\}$, and $\{\alpha_{j,3}^k\}$;
\item each $T_j$ is a parallel module with children $\{t_{j,1}\}$, $\{t_{j,2}\}$, $\{t_{j,3}\}$, and $\{t_{j,4}\}$.
\end{itemize}

Now, we construct a partial representation $\phi:V(G)\to\calF^\circ$ of $G$ (see Fig.~\ref{fig:clause}).
For a set $S\subset V(G)$, let $\dom\phi(S)=\bigcup_{u\in S}\dom\phi(u)$.
For $u\in X$, we choose $\dom\phi(u)$ to be intervals of equal length such that
\begin{itemize}
\item $\dom\phi(x_{i,1}^+)\cap\dom\phi(x_{i,2}^+)$ and $\dom\phi(x_{i,2}^+)\cap\dom\phi(x_{i,3}^+)$ are singletons, and $\dom\phi(x_{i,1}^+)$ is entirely to the left of $\dom\phi(x_{i,3}^+)$;
\item $\dom\phi(x_{i,1}^-)\cap\dom\phi(x_{i,2}^-)$ and $\dom\phi(x_{i,2}^-)\cap\dom\phi(x_{i,3}^-)$ are singletons, and $\dom\phi(x_{i,1}^-)$ is entirely to the left of $\dom\phi(x_{i,3}^-)$;
\item the intervals $\dom\phi(X_1^+),\dom\phi(X_1^-),\ldots,\dom\phi(X_n^+),\dom\phi(X_n^-)$ are pairwise disjoint and contained in $(0,1)$.
\end{itemize}
Next, we define $\dom\phi(u)$ for $u\in V(G)-X$ as follows:
\begin{itemize}
\item if $\alpha_j^k=x_i$, then $\dom\phi(\alpha_{j,r}^k)=\dom\phi(x_{i,r}^+)$ for $r\in\{1,2,3\}$;
\item if $\alpha_j^k=\neg x_i$, then $\dom\phi(\alpha_{j,r}^k)=\dom\phi(x_{i,r}^-)$ for $r\in\{1,2,3\}$;
\item $\dom\phi(t_{j,1})$, $\dom\phi(t_{j,2})$, $\dom\phi(t_{j,3})$, and $\dom\phi(t_{j,4})$ are the closures of the four open intervals that the set $(0,1)-\dom\phi(A_j^1\cup A_j^2\cup A_j^3)$ splits into;
\item $\dom\phi(p)=\dom\phi(q)=[0,1]$.
\end{itemize}
Now, we define all functions $\phi(u)$ for $u\in V(G)$ to be constant, as follows:
\begin{itemize}
\item $\phi(u)\equiv 0$ for $u\in X$;
\item $\phi(p)\equiv 1$;
\item $\phi(q)\equiv 2$;
\item $\phi(u)\equiv 3$ for $u\in T_1\cup\ldots\cup T_m$;
\item $\phi(u)\equiv 4$ for $u\in A_1^1\cup A_1^2\cup A_1^3\cup\ldots\cup A_m^1\cup A_m^2\cup A_m^3$.
\end{itemize}
We can easily check that $\phi$ is indeed a partial representation of $G$.

It remains to prove that $\alpha$ is satisfiable if and only if $\phi$ is extendable.
To this end, we apply Lemma \ref{lem:regions-partial} and prove that $\alpha$ is satisfiable if and only if $G$ has a transitive orientation $P$ that respects $\phi$ and satisfies $\Reg_P(u)\cap\Reg_P(v)\neq\emptyset$ for any two incomparable vertices $u$ and $v$.

We know from Theorem \ref{thm:correspondence} that all transitive orientations of $G$ are obtained by transitively orienting the strong modules of $G$.
All transitive orientations of the strong modules of $G$ respecting $\phi$ have the following in common (see Fig.\ \ref{fig:graph}):
\begin{itemize}
\item the module $V(G)$ is oriented so that $\{p\}<_PA_j$, $\{p\}<_P\{q\}$, and $X_i<_P\{q\}$;
\item each module $A_j$ is oriented so that $T_j<_PA_j^1,A_j^2,A_j^3$.
\end{itemize}
Thus choosing a transitive orientation of $G$ respecting $\phi$ is equivalent to fixing the order of $A_j^1,A_j^2,A_j^3$ in every $A_j$ and choosing one of the two orientations of every $X_j$.
Figures \ref{fig:region-clause} and \ref{fig:region-variable} illustrate how the regions of the vertices depend on the chosen orientation.

Suppose that $\alpha$ is satisfied by a valuation $\Gamma$ of the variables.
This means that in each clause $\alpha_j$ at least one of the literals $\alpha_j^1,\alpha_j^2,\alpha_j^3$ is true in $\Gamma$.
We choose transitive orientations of the strong modules of $G$ (and thus a transtive orientation $P$ of $G$) so that they respect $\phi$ and satisfy the following:
\begin{itemize}
\item if $x_i$ is true in $\Gamma$, then $X_i^+<_PX_i^-$, otherwise $X_i^-<_PX_i^+$;
\item $A_j^{k_j}>_PA_j^\ell$ for $\ell\neq k_j$, where $\alpha_j^{k_j}$ is a literal that is true in $\Gamma$.
\end{itemize}
The region $\Reg_P(x_{i,r}^+)$ or $\Reg_P(x_{i,r}^-)$ has non-empty intersection with $\Reg_P(\alpha_{j,s}^k)$ inside $\dom\phi(\alpha_{j,t}^{k_j})$, for $\{r,s,t\}=\{1,2,3\}$ and any $k$.
The regions $\Reg_P(u)$ and $\Reg_P(v)$ for other choices of $u$ and $v$ intersect regardless of which transitive orientation $P$ we choose.

Now, suppose that $P$ is a transitive orientation of $G$ with the property that $\Reg_P(u)\cap\Reg_P(v)\neq\emptyset$ for any two incomparable vertices $u$ and $v$.
We define a valuation $\Gamma$ as follows: $x_i$ is true if and only if $X_i^+<_PX_i^-$.
Let $A_j^k$ be the greatest of the modules $A_j^1,A_j^2,A_j^3$ with respect to $<_P$.
We show that the literal $\alpha_j^k$ is true in $\Gamma$.
The regions $\Reg_P(\alpha_{j,*}^k)$ go below the horizontal line $\phi(q)$ only inside $\dom\phi(A_j^k)\times\setR$, elsewhere they are blocked from below by functions $\phi(t_{j,*})$ and $\phi(\alpha_{j,*}^\ell)$ for $\ell\neq k$.
Thus $\Reg_P(\alpha_{j,*}^k)$ can intersect $\Reg_P(x_{i,*}^+)$ or $\Reg_P(x_{i,*}^-)$ only inside $\dom\phi(A_j^k)\times\setR$.
Suppose $\alpha_j^k=x_i$.
It follows that $\dom\phi(A_j^k)=\dom\phi(X_i^+)$.
If $X_i^-<_PX_i^+$, then the parts of $\Reg_P(x_{i,*}^-)$ inside $\dom\phi(X_i^+)$ are blocked from above by functions $\phi(x_{i,*}^+)$ and thus lie entirely below the horizontal line $\phi(p)$, which contradicts the assumption that $\Reg_P(\alpha_{j,*}^k)\cap\Reg_P(x_{i,*}^-)\neq\emptyset$.
Thus we have $X_i^+<_PX_i^-$, which means that the literal $\alpha_j^k=x_i$ is true in $\Gamma$.
An analogous argument shows that if $\alpha_j^k=\neg x_i$ then $\alpha_j^k$ is true in $\Gamma$.
This completes the proof.

\end{document}

%% file: introduction.tex
\newcommand{\heading}[1]{\medbreak\par\noindent\hbox{\bf #1}\enspace}
\def\co{\hbox{\sf co}}
\def\comp{\hbox{\sf CO}}
\def\perm{\hbox{\sf PERM}}
\def\fun{\hbox{\sf FUN}}
\def\parfun{\hbox{\sf PartialFUN}}
\def\ext{\textsc{RepExt}}
\def\orientext{\textsc{OrientExt}}

\def\computationproblem#1#2#3{
	\begin{center}
	\begin{tabular}{|rp{10cm}|}
	\hline
	{\bf Problem:\enspace}&#1\\
	{\bf Input:\enspace}&#2\\
	{\bf Output:\enspace}&#3\\
	\hline
	\end{tabular}
	\end{center}
}

\section{Introduction}

Geometric representations of graphs have been studied as part of graph theory from its very beginning.
Euler initiated the study of graph theory by studying planar graphs in the setting of three-dimensional polytopes.
The theorem of Kuratowski~\cite{kuratowski_minors} provides the first combinatorial characterization of planar graphs and can be considered as the start of modern graph theory.

In this paper we are interested in intersection representations, which assign geometric objects to the vertices of graphs and the edges are encoded by intersections of objects.
Formally, an intersection representation of $G$ is a mapping $\phi:V(G)\to\calS$ of the vertices of $G$ to a class $\calS$ of objects (sets) such that $\phi(u)\cap\phi(v)\neq\emptyset$ if and only if $uv\in E(G)$.
This way, for different classes $\calS$ we obtain various classes of representable graphs.
Classic examples include interval graphs, circle graphs, permutation graphs, string graphs, convex graphs, and function graphs~\cite{Golumbic2,egr}. As seen from these two monographs, geometric intersection graphs are intensively studied for their applied motivation, algorithmic properties, but also as a source of many interesting theoretical results that sometimes stimulate our curiosity (as one example we mention that string graphs requiring exponential number of crossing points in every representation are known, and yet the recognition problem is in NP).

Naturally, the recognition problem is the first one to consider. For most of the intersection defined classes the complexity of their recognition is known. 
For example, interval graphs can be recognized in linear time~\cite{boothlueker}, while recognition of string graphs is NP-complete~\cite{krat_string,sss_string}. Our goal is to study the easily recognizable classes and explore if 
the recognition problem becomes harder when extra conditions are given with the input.

\heading{Partial representations.}
A recent paper of Klav\'ik et al.~\cite{kkv} introduced a question of extending partial representations for classes of intersection graphs.
A \emph{partial representation} of $G$ is a representation $\phi:R\to\calS$ of the induced subgraph $G[R]$ for a set $R\subset V(G)$.
The problem $\ext(\calG)$ of partial representation extension for a class $\calG$ represented in a class $\calS$ is defined as follows: given a graph $G\in\calG$ and a partial representation $\phi:R\to\calS$ of $G$, decide whether there is a representation $\psi:V(G)\to\calS$ that \emph{extends} $\phi$, that is, such that $\psi|_R=\phi$.


The paper~\cite{kkv} investigates the complexity of the problem for intervals graphs (intersection graphs of intervals on a line) and presents an $O(n^2)$ algorithm for extending interval representations and an $O(nm)$ algorithm for extending proper interval representations of graphs with $n$ vertices and $m$ edges.
A recent result of Bl\"asius and Rutter~\cite{ignaz} solves the problem of extending interval representations in time $O(n+m)$, but the algorithm is involved.

A related problem of simultaneous graph representations was recently introduced by Jampani and Lubiw~\cite{jampani}: given two graphs $G$ and $H$ sharing common vertices $I=V(G)\cap V(H)$, decide whether there are representations $\phi$ of $G$ and $\psi$ of $H$ such that $\phi|_I=\psi|_I$.
Simultaneous representations are closely related to partial representation extension.
Namely, in many cases we can solve partial representation extension by introducing an additional graph and putting $I=R$.
On the other hand, if $|I|$ is small, then we can test all essentially different possible representations of $I$ and try to extend them to $V(G)$ and $V(H)$, which can give us a polynomial-time algorithm for fixed parameter $|I|$.

Several other problems have been considered in which a partial solution is given and the task is to extend it.
For example, every $k$-regular bipartite graph is $k$-edge-colorable, but if some edges are pre-colored, the extension problem becomes NP-complete even for $k=3$~\cite{fiala}, and even when the input is restricted to planar graphs~\cite{marx}.
For planar graphs, partial representation extension is solvable in linear time~\cite{angelini}.
Every planar graph admits a straight-line drawing, but extending such representation is NP-complete~\cite{patrigani}.

\heading{Permutation and function graphs.}
In this paper, we consider two classes of intersection graphs.
The class \fun\ of function graphs is represented by continuous monotone curves connecting two parallel lines in the plane.
In other words, a representation of a function graph assigns a continuous function $f:[0,1]\to\setR$ to every vertex of the graph.
The class \perm\ of permutation graphs, which is a subclass of \fun, is represented the same way by linear functions.

A graph is a \emph{comparability graph} if it is possible to orient its edges transitively.
An orientation is \emph{transitive} if $u\to v$ and $v\to w$ imply $u\to w$.
Thus the relation $\to$ in a transitively oriented graph is a strict partial order.
We denote the class of comparability graphs by \comp.
A \emph{partial orientation} of a comparability graph is a transitive orientation of some of its edges.
The problem $\orientext$ is to decide whether we can orient the remaining edges to get a transitive orientation of the entire graph.

By \co\comp\ we denote the class of complements of comparability graphs.
We have the following relations: $\fun=\co\comp$ \cite{GRU} and $\perm=\comp\cap\co\comp$ \cite{perm}.
We derive a transitive ordering from a function graphs as follows: if two functions do not intersect, then one is on top of the other; thus we can order the functions from bottom to top.
For permutation graphs, we use the fact that $\perm=\co\perm$.

\heading{Our results.}
By a straightforward modification of the recognition algorithms of Golumbic \cite{Golumbic1,Golumbic2} and by the property $\perm=\comp\cap\co\comp$, we get the following.

\begin{proposition}\label{prop:comp_ext}
The problem\/ $\orientext$ can be solved in time\/ $O((n+m)\Delta)$ for graphs with\/ $n$ vertices, $m$\/ edges, and maximum degree\/ $\Delta$.
\end{proposition}

\begin{proposition}\label{prop:perm_ext}
The problem\/ $\ext(\perm)$ can be solved in time\/ $O(n^3)$ for graphs with\/ $n$ vertices.
\end{proposition}

Our first main result is a polynomial-time algorithm for $\ext(\fun)$.
Here the straightforward generalization of the recognition algorithm does not work.
Even though $\fun=\co\comp$, the problems $\ext(\fun)$ and $\orientext$ are different, see Fig.~\ref{fun_is_not_coco}.
This is similar to what happens for the classes of proper and unit interval graphs: they are known to be equal, but their partial representation extension problems are different~\cite{kkv}.

\begin{theorem}\label{thm:fun_ext}
The problem\/ $\ext(\fun)$ can be solved in polynomial time.
\end{theorem}

\begin{figure}[tb]
\centering\input{intro.tex}
\caption{A function graph $G$ with a partial representation that is not extendable: $\phi(u)$ in order to intersect $\phi(b)$ and $\phi(d)$ must also intersect $\phi(a)$ or $\phi(c)$.
The corresponding partial orientation of the comparability graph $\overline G$ is extendable.}
\label{fun_is_not_coco}
\end{figure}

The second main result concerns partial representations by \emph{partial functions} $f:[a,b]\to\setR$ with $[a,b]\subset[0,1]$, which generalize ordinary partial representations by functions.
The problem $\ext^\star(\fun)$ is to decide, for a given graph, whether a given partial representation by partial functions can be extended to a representation of the whole graph so that all partial functions are extended to functions defined on the entire $[0,1]$.

\begin{theorem}\label{thm:partial_fun_ext}
The problem\/ $\ext^\star(\fun)$ is NP-complete.
\end{theorem}

%% file: intro.tex
\ifx\JPicScale\undefined\def\JPicScale{1}\fi
\psset{unit=\JPicScale mm}
\psset{linewidth=0.3,dotsep=1,hatchwidth=0.3,hatchsep=1.5,shadowsize=1,dimen=middle}
\psset{dotsize=0.7 2.5,dotscale=1 1,fillcolor=black}
\psset{arrowsize=1 2,arrowlength=1,arrowinset=0.25,tbarsize=0.7 5,bracketlength=0.15,rbracketlength=0.15}
\begin{pspicture}(0,0)(86,23.2)
\rput{0}(6,4){\psellipse[linestyle=none,fillstyle=solid](0,0)(1,1)}
\rput{0}(18,12){\psellipse[linestyle=none,fillstyle=solid](0,0)(1,1)}
\rput{0}(11,20){\psellipse[linestyle=none,fillstyle=solid](0,0)(1,1)}
\rput{0}(3,12){\psellipse[linestyle=none,fillstyle=solid](0,0)(1,1)}
\rput(21,12){$a$}
\rput(8,20){$b$}
\rput(0,12){$c$}
\rput(18,4){$u$}
\rput(28,21){$a$}
\rput(28,16){$b$}
\rput(28,9){$c$}
\rput(3,4){$d$}
\rput(17,16.5){}
\rput(28,4){$d$}
\psline(30,0)(30,23.2)
\psline(58,0)(58,23.2)
\psline[linecolor=white](0,0)(0,2)
\psline(3,12)(11,20)
\psline(11,20)(18,12)
\psline(18,12)(3,12)
\psline(6,4)(15,4)
\psline(15,4)(11,20)
\newrgbcolor{userFillColour}{0.6 0.6 0.6}
\rput{0}(15,4){\psellipse[linestyle=none,fillcolor=userFillColour,fillstyle=solid](0,0)(1,1)}
\rput{0}(71,4){\psellipse[linestyle=none,fillstyle=solid](0,0)(1,1)}
\rput{0}(83,12){\psellipse[linestyle=none,fillstyle=solid](0,0)(1,1)}
\rput{0}(76,20){\psellipse[linestyle=none,fillstyle=solid](0,0)(1,1)}
\rput{0}(68,12){\psellipse[linestyle=none,fillstyle=solid](0,0)(1,1)}
\rput(86,12){$a$}
\rput(73,20){$b$}
\rput(83,4){$u$}
\rput(68,4){$d$}
\rput(82,16.5){}
\newrgbcolor{userFillColour}{0.6 0.6 0.6}
\rput{0}(80,4){\psellipse[linestyle=none,fillcolor=userFillColour,fillstyle=solid](0,0)(1,1)}
\rput(65,12){$c$}
\psline{->}(70.4,5.4)(68.4,10.6)
\psline{->}(71.4,5.4)(75.6,18.6)
\psline{->}(72.2,4.8)(82,11)
\newrgbcolor{userLineColour}{0.6 0.6 0.6}
\psline[linecolor=userLineColour]{->}(78.8,4.6)(69.2,11.2)
\newrgbcolor{userLineColour}{0.6 0.6 0.6}
\psline[linecolor=userLineColour]{->}(80.4,5.2)(82.6,10.6)
\pscustom[linewidth=0.25]{\psbezier(30,21)(33.5,21)(35.6,19.8)(37,17)
\psbezier(38.4,14.2)(39.75,12.4)(41.5,11)
\psbezier(43.25,9.6)(44.75,9.6)(46.5,11)
\psbezier(48.25,12.4)(49.6,14.35)(51,17.5)
\psbezier(52.4,20.65)(53.75,22)(55.5,22)
\psbezier(57.25,22)(58,22)(58,22)
}
\pscustom[linewidth=0.25]{\psbezier(30,9)(30,9)(30.75,9)(32.5,9)
\psbezier(34.25,9)(35.6,10.35)(37,13.5)
\psbezier(38.4,16.65)(39.75,18.45)(41.5,19.5)
\psbezier(43.25,20.55)(44.6,20.55)(46,19.5)
\psbezier(47.4,18.45)(48.75,16.65)(50.5,13.5)
\psbezier(52.25,10.35)(53.75,9)(55.5,9)
\psbezier(57.25,9)(58,9)(58,9)
}
\pscustom[linewidth=0.25]{\psbezier(30,16)(32.8,15.65)(34.75,16.32)(36.5,18.25)
\psbezier(38.25,20.18)(39.75,19.65)(41.5,16.5)
\psbezier(43.25,13.35)(44.75,13.35)(46.5,16.5)
\psbezier(48.25,19.65)(49.75,20.18)(51.5,18.25)
\psbezier(53.25,16.32)(54.6,15.58)(56,15.75)
\psbezier(57.4,15.92)(58,16)(58,16)
}
\psline[linewidth=0.25](30,4)(58,4)
\newrgbcolor{userLineColour}{0.6 0.6 0.6}
\pscustom[linewidth=0.25,linecolor=userLineColour]{\psbezier(30,1)(32.8,1.7)(34.75,3.35)(36.5,6.5)
\psbezier(38.25,9.65)(39.75,11.9)(41.5,14)
\psbezier(43.25,16.1)(44.75,16.1)(46.5,14)
\psbezier(48.25,11.9)(49.75,9.65)(51.5,6.5)
\psbezier(53.25,3.35)(54.6,1.85)(56,1.5)
\psbezier(57.4,1.15)(58,1)(58,1)
}
\rput(60,1){$u$}
\rput(53,1.5){$\mathbf{?}$}
\end{pspicture}

%% file: comp_and_perm.tex
\section{Extending partial orientation of a comparability graph}\label{sec:comparability}

We show how to modify the recognition algorithm of Golumbic~\cite{Golumbic1,Golumbic2} to obtain an algorithm for extending partial orientation of a comparability graph.
The recognition algorithm repeats the following until the whole graph $G$ is oriented.
We pick an arbitrary unoriented edge and orient it in any direction.
This may force several other edges to be oriented, according to the following rules:
\begin{itemize}
\item $u\to v$ and $v\to w$ force $u\to w$,
\item $u\to v$, $vw\in E(G)$, and $uw\notin E(G)$ force $w\to v$.
\end{itemize}
If we find an edge that we are forced to reorient (change its direction), then we stop and answer that $G$ is not a comparability graph.
Otherwise, we finally orient all edges and obtain a transitive orientation of $G$.
The running time of the algorithm is $O((n+m)\Delta)$.

Now, we adapt this algorithm to the problem $\orientext$.
Since the algorithm processes edges in an arbitrary order, we can choose an ordering $e_1<\ldots<e_m$ of the edges and always pick the first non-oriented edge in this ordering.
Suppose that the first $k$ edges $e_1,\dots,e_k$ are preoriented by $\phi$.
If we pick an edge $e_i$, then we orient it according to $\phi$ if $i\leq k$ or arbitrarily otherwise.
The algorithm additionally fails if it is forced to orient an edge $e_i$ with $i\leq k$ in the opposite
direction to the one forced by $\phi$.
In such a case, this orientation is forced by the orientation of $e_1,\dots,e_{i-1}$, and thus the partial orientation is indeed not extendible.
The running time of the algorithm is again $O((n+m)\Delta)$, which proves Proposition~\ref{prop:comp_ext}.
The detailed proof of Proposition \ref{prop:perm_ext} is in the Appendix.

%% file: function_graphs.tex
\ifx\JPicScale\undefined\def\JPicScale{1}\fi
\psset{unit=\JPicScale mm}
\psset{linewidth=0.3,dotsep=1,hatchwidth=0.3,hatchsep=1.5,shadowsize=1,dimen=middle}
\psset{dotsize=0.7 2.5,dotscale=1 1,fillcolor=black}
\psset{arrowsize=1 2,arrowlength=1,arrowinset=0.25,tbarsize=0.7 5,bracketlength=0.15,rbracketlength=0.15}
\begin{pspicture}(0,0)(90,26)
\newrgbcolor{userLineColour}{0.8 0.8 0.8}
\newrgbcolor{userFillColour}{0.8 0.8 0.8}
\pscustom[linecolor=userLineColour,fillcolor=userFillColour,fillstyle=solid]{\psline(35.5,9)(35.6,22.6)
\psline(35.6,22.6)(38.8,18.8)
\psline(38.8,18.8)(40.5,17)
\psline(40.5,17)(55.4,24.8)
\psline(55.4,24.8)(55.5,19)
\psline(55.5,19)(55.5,12)
\psline(55.5,12)(51.5,15)
\psline(51.5,15)(45.5,9)
\psline(45.5,9)(40.5,12)
\psline(40.5,12)(35.5,9)
\psbezier(35.5,9)(35.5,9)(35.5,9)
\psbezier(35.5,9)(35.5,9)(35.5,9)
\closepath}
\psline[linewidth=0.2]{->}(6.8,5)(18,16)
\psline[linewidth=0.2]{<-}(19,15.8)(19,5.4)
\psline[linewidth=0.2]{<-}(9.4,19.6)(10,15.4)
\psline[linewidth=0.2]{->}(20,5)(25,10)
\psline[linewidth=0.2]{->}(6.4,5.2)(9.6,12.6)
\rput{0}(6,4){\psellipse[linestyle=none,fillstyle=solid](0,0)(1,1)}
\newrgbcolor{userFillColour}{0.6 0.6 0.6}
\rput{0}(26,11){\psellipse[linestyle=none,fillcolor=userFillColour,fillstyle=solid](0,0)(1,1)}
\rput{0}(19,17){\psellipse[linestyle=none,fillstyle=solid](0,0)(1,1)}
\rput{0}(9.2,20.8){\psellipse[linestyle=none,fillstyle=solid](0,0)(1,1)}
\psline[linewidth=0.2]{<-}(10.8,13)(18,5)
\newrgbcolor{userFillColour}{0.6 0.6 0.6}
\rput{0}(10,14){\psellipse[linestyle=none,fillcolor=userFillColour,fillstyle=solid](0,0)(1,1)}
\rput{0}(19,4){\psellipse[linestyle=none,fillstyle=solid](0,0)(1,1)}
\rput(3.5,4){$a$}
\rput(16,4){$b$}
\rput(6.8,21){$c$}
\rput(23.5,11){$u$}
\rput(12,14){$v$}
\rput(34,6){$a$}
\rput(34,9){$b$}
\rput(34,23){$c$}
\rput(16.4,17.4){$d$}
\rput(17,16.5){}
\rput(34,17){$d$}
\psline[linewidth=0.1,linestyle=dashed,dash=1 1](40.5,2.8)(40.5,26)
\rput(41,1){$x_1$}
\rput(48.5,1){$x_2$}
\rput(52,1){$x_3$}
\rput(56,1){$x_4$}
\rput(36,1){$x_0$}
\pscustom[]{\psline(35.5,5.8)(40.5,3.8)
\psline(40.5,3.8)(51.5,14.8)
\psline(51.5,14.8)(55.5,11.8)
\psbezier(55.5,11.8)(55.5,11.8)(55.5,11.8)
}
\pscustom[]{\psline(35.5,8.8)(40.5,11.8)
\psline(40.5,11.8)(51.5,4.8)
\psline(51.5,4.8)(55.5,6.8)
\psbezier(55.5,6.8)(55.5,6.8)(55.5,6.8)
}
\psline(35.5,16.8)
(48,24.6)(55.5,18.8)
\pscustom[]{\psline(35.5,22.8)(40.5,16.8)
\psline(40.5,16.8)(55.5,24.8)
\psbezier(55.5,24.8)(55.5,24.8)(55.5,24.8)
}
\newrgbcolor{userLineColour}{0.6 0.6 0.6}
\psline[linecolor=userLineColour](35.5,12.3)
(46,12.2)
(51.4,20.6)
(55.4,22.6)(55.5,16.8)
\psline(35.5,2.8)(35.5,26)
\psline(55.5,2.8)(55.5,26)
\newrgbcolor{userLineColour}{0.6 0.6 0.6}
\rput(42,14){}
\newrgbcolor{userLineColour}{0.6 0.6 0.6}
\rput(44.4,15.4){\tiny{$\Reg(v)$}}
\newrgbcolor{userLineColour}{0.6 0.6 0.6}
\rput(57.5,22.5){$v$}
\psline[linewidth=0.1,linestyle=dashed,dash=1 1](48,2.8)(48,26)
\psline[linewidth=0.1,linestyle=dashed,dash=1 1](51.5,2.8)(51.5,26)
\rput(66.5,6){$a$}
\rput(66.5,9){$b$}
\rput(66.5,23){$c$}
\rput(66.5,17){$d$}
\psline[linewidth=0.1,linestyle=dashed,dash=1 1](73,2.8)(73,26)
\rput(73.5,1){$x_1$}
\rput(81,1){$x_2$}
\rput(84.5,1){$x_3$}
\rput(88.5,1){$x_4$}
\rput(68.5,1){$x_0$}
\pscustom[]{\psline(68,5.8)(73,3.8)
\psline(73,3.8)(84,14.8)
\psline(84,14.8)(88,11.8)
\psbezier(88,11.8)(88,11.8)(88,11.8)
}
\pscustom[]{\psline(68,8.8)(73,11.8)
\psline(73,11.8)(84,4.8)
\psline(84,4.8)(88,6.8)
\psbezier(88,6.8)(88,6.8)(88,6.8)
}
\psline(68,16.8)
(80.5,24.6)(88,18.8)
\pscustom[]{\psline(68,22.8)(73,16.8)
\psline(73,16.8)(88,24.8)
\psbezier(88,24.8)(88,24.8)(88,24.8)
}
\newrgbcolor{userLineColour}{0.6 0.6 0.6}
\psline[linecolor=userLineColour](68,12.3)
(78.5,12.2)
(83.9,20.6)
(87.9,22.6)(88,16.8)
\psline(88,2.8)(88,26)
\newrgbcolor{userLineColour}{0.6 0.6 0.6}
\rput(74.5,14){}
\newrgbcolor{userLineColour}{0.6 0.6 0.6}
\rput(90,22.5){$v$}
\psline[linewidth=0.1,linestyle=dashed,dash=1 1](80.5,2.8)(80.5,26)
\psline[linewidth=0.1,linestyle=dashed,dash=1 1](84,2.8)(84,26)
\newrgbcolor{userLineColour}{0.6 0.6 0.6}
\pscustom[linecolor=userLineColour]{\psline(88,9)(82.5,10.5)
\psline(82.5,10.5)(68,26)
\psbezier(68,26)(68,26)(68,26)
}
\newrgbcolor{userLineColour}{0.6 0.6 0.6}
\rput(90,9){$u$}
\psline[linecolor=white](0,0)(0,2)
\psline(68,2.8)(68,26)
\psline[linewidth=0.2]{->}(6,5.5)(8.6,19.6)
\psline[linewidth=0.2]{->}(18.6,5.2)(10.2,19.8)
\end{pspicture}

%% file: graph.tex
\ifx\JPicScale\undefined\def\JPicScale{1}\fi
\psset{unit=\JPicScale mm}
\psset{linewidth=0.3,dotsep=1,hatchwidth=0.3,hatchsep=1.5,shadowsize=1,dimen=middle}
\psset{dotsize=0.7 2.5,dotscale=1 1,fillcolor=black}
\psset{arrowsize=1 2,arrowlength=1,arrowinset=0.25,tbarsize=0.7 5,bracketlength=0.15,rbracketlength=0.15}
\begin{pspicture}(0,0)(120.5,60.12)
\rput(70.5,18.5){}
\rput{0}(56.75,2.5){\psellipse[linestyle=none,fillstyle=solid](0,0)(1,1)}
\psline[linecolor=white](0,0)(0,7)
\psline[linecolor=white](112.5,60)(120.5,60)
\rput{0}(94,55){\psellipse[linestyle=none,fillstyle=solid](0,0)(1,1)}
\rput{90}(17,40){\psellipse[](0,0)(20.12,-17)}
\rput(3,58){\Large{$A_1$}}
\rput{0}(17,27){\psellipse[](0,0)(10,-3)}
\rput(17,22){$T_1$}
\rput{0}(8,41){\psellipse[](0,0)(4,-4)}
\rput{0}(11,27){\psellipse[linestyle=none,fillstyle=solid](0,0)(1,1)}
\rput{0}(15,27){\psellipse[linestyle=none,fillstyle=solid](0,0)(1,1)}
\rput{0}(19,27){\psellipse[linestyle=none,fillstyle=solid](0,0)(1,1)}
\rput{0}(23,27){\psellipse[linestyle=none,fillstyle=solid](0,0)(1,1)}
\rput{0}(6,40){\psellipse[linestyle=none,fillstyle=solid](0,0)(1,1)}
\rput{0}(10,40){\psellipse[linestyle=none,fillstyle=solid](0,0)(1,1)}
\rput{0}(8,43){\psellipse[linestyle=none,fillstyle=solid](0,0)(1,1)}
\rput{0}(17,51){\psellipse[](0,0)(4,-4)}
\rput{0}(15,50){\psellipse[linestyle=none,fillstyle=solid](0,0)(1,1)}
\rput{0}(19,50){\psellipse[linestyle=none,fillstyle=solid](0,0)(1,1)}
\rput{0}(17,53){\psellipse[linestyle=none,fillstyle=solid](0,0)(1,1)}
\rput{0}(26,41){\psellipse[](0,0)(4,-4)}
\rput{0}(24,40){\psellipse[linestyle=none,fillstyle=solid](0,0)(1,1)}
\rput{0}(28,40){\psellipse[linestyle=none,fillstyle=solid](0,0)(1,1)}
\rput{0}(26,43){\psellipse[linestyle=none,fillstyle=solid](0,0)(1,1)}
\rput(3.5,36){$A^1_1$}
\rput(31,36.5){$A^2_1$}
\rput(10,51){$A_1^3$}
\rput(2,38){}
\psline[linewidth=0.15]{->}(11,30)(7,36)
\psline[linewidth=0.15]{->}(11,30)(11,37.5)
\psline[linewidth=0.15]{->}(14.5,30.5)(15.5,46.5)
\psline[linewidth=0.15]{->}(19.5,30.5)(18.5,46.5)
\psline[linewidth=0.15]{->}(11,30)(22.5,37.5)
\psline[linewidth=0.15]{->}(23,30)(26,36.5)
\psline[linewidth=0.15]{->}(11,30)(25.5,36.5)
\psline[linewidth=0.15]{->}(23,30)(23,37.5)
\psline[linewidth=0.15]{->}(23,30)(11.5,37.5)
\psline[linewidth=0.15]{->}(23,30)(8,36)
\psline[linewidth=0.15]{->}(14.5,30.5)(18,46)
\psline[linewidth=0.15]{->}(19.5,30.5)(16,46)
\psline[linewidth=0.4](7.5,45.5)(13,49)
\psline[linewidth=0.4](11.5,44)(15,47)
\psline[linewidth=0.4](11.5,44)(13,49)
\psline[linewidth=0.4](7.5,45.5)(15,47)
\psline(19,47)(22.5,44)
\psline[linewidth=0.4](19,47)(26.5,45.5)
\psline[linewidth=0.4](26.5,45.5)(21,49)
\psline[linewidth=0.4](21,49)(22.5,44)
\psline[linewidth=0.4](12.5,42)(21.5,42)
\psline[linewidth=0.4](12.5,39.5)(21.5,39.5)
\psline[linewidth=0.4](21.5,39.5)(12.5,42)
\psline[linewidth=0.4](12.5,39.5)(21.5,42)
\rput{90}(57,40){\psellipse[](0,0)(20.12,-17)}
\rput(43,58){\Large{$A_m$}}
\rput{0}(57,27){\psellipse[](0,0)(10,-3)}
\rput(57,22){$T_m$}
\rput{0}(48,41){\psellipse[](0,0)(4,-4)}
\rput{0}(51,27){\psellipse[linestyle=none,fillstyle=solid](0,0)(1,1)}
\rput{0}(55,27){\psellipse[linestyle=none,fillstyle=solid](0,0)(1,1)}
\rput{0}(59,27){\psellipse[linestyle=none,fillstyle=solid](0,0)(1,1)}
\rput{0}(63,27){\psellipse[linestyle=none,fillstyle=solid](0,0)(1,1)}
\rput{0}(46,40){\psellipse[linestyle=none,fillstyle=solid](0,0)(1,1)}
\rput{0}(50,40){\psellipse[linestyle=none,fillstyle=solid](0,0)(1,1)}
\rput{0}(48,43){\psellipse[linestyle=none,fillstyle=solid](0,0)(1,1)}
\rput{0}(57,51){\psellipse[](0,0)(4,-4)}
\rput{0}(55,50){\psellipse[linestyle=none,fillstyle=solid](0,0)(1,1)}
\rput{0}(59,50){\psellipse[linestyle=none,fillstyle=solid](0,0)(1,1)}
\rput{0}(57,53){\psellipse[linestyle=none,fillstyle=solid](0,0)(1,1)}
\rput{0}(66,41){\psellipse[](0,0)(4,-4)}
\rput{0}(64,40){\psellipse[linestyle=none,fillstyle=solid](0,0)(1,1)}
\rput{0}(68,40){\psellipse[linestyle=none,fillstyle=solid](0,0)(1,1)}
\rput{0}(66,43){\psellipse[linestyle=none,fillstyle=solid](0,0)(1,1)}
\rput(43.5,36){$A^1_m$}
\rput(70,36){$A^2_m$}
\rput(50,51){$A_m^3$}
\rput(42,38){}
\psline[linewidth=0.15]{->}(51,30)(47,36)
\psline[linewidth=0.15]{->}(51,30)(51,37.5)
\psline[linewidth=0.15]{->}(54.5,30.5)(55.5,46.5)
\psline[linewidth=0.15]{->}(59.5,30.5)(58.5,46.5)
\psline[linewidth=0.15]{->}(51,30)(62.5,37.5)
\psline[linewidth=0.15]{->}(63,30)(66,36.5)
\psline[linewidth=0.15]{->}(51,30)(65.5,36.5)
\psline[linewidth=0.15]{->}(63,30)(63,37.5)
\psline[linewidth=0.15]{->}(63,30)(51.5,37.5)
\psline[linewidth=0.15]{->}(63,30)(48,36)
\psline[linewidth=0.15]{->}(54.5,30.5)(58,46)
\psline[linewidth=0.15]{->}(59.5,30.5)(56,46)
\psline[linewidth=0.4](47.5,45.5)(53,49)
\psline[linewidth=0.4](51.5,44)(55,47)
\psline[linewidth=0.4](51.5,44)(53,49)
\psline[linewidth=0.4](47.5,45.5)(55,47)
\psline(59,47)(62.5,44)
\psline[linewidth=0.4](59,47)(66.5,45.5)
\psline[linewidth=0.4](66.5,45.5)(61,49)
\psline[linewidth=0.4](61,49)(62.5,44)
\psline[linewidth=0.4](52.5,42)(61.5,42)
\psline[linewidth=0.4](52.5,39.5)(61.5,39.5)
\psline[linewidth=0.4](61.5,39.5)(52.5,42)
\psline[linewidth=0.4](52.5,39.5)(61.5,42)
\rput(37,40){$\ldots$}
\psline[linewidth=0.15]{->}(58,3.5)(93,54)
\psline[linewidth=0.15]{->}(57.5,4)(65,19)
\psline[linewidth=0.15]{->}(57,4)(59,19)
\psline[linewidth=0.15]{->}(56.5,4)(53,19)
\psline[linewidth=0.15]{->}(56,3.75)(41.5,19)
\psline[linewidth=0.15]{->}(55,3)(14,19)
\psline[linewidth=0.15]{->}(55.5,3.5)(29,19)
\rput(74.5,16){}
\psline[linewidth=0.15]{<-}(93.5,53.5)(81,28)
\psline[linewidth=0.15]{<-}(94,53.5)(89,28)
\psline[linewidth=0.15]{<-}(94.5,53.5)(97,28)
\psline[linewidth=0.15]{<-}(95,54)(106,28)
\psline[linewidth=0.15]{<-}(95.5,54.5)(115,28)
\rput(100,18.5){}
\rput{90}(109.5,13.5){\psellipse[](0,0)(13.5,-10.5)}
\rput{0}(108,8){\psellipse[](0,0)(7,-3)}
\rput{0}(104,8){\psellipse[linestyle=none,fillstyle=solid](0,0)(1,1)}
\rput{0}(108,8){\psellipse[linestyle=none,fillstyle=solid](0,0)(1,1)}
\rput{0}(112,8){\psellipse[linestyle=none,fillstyle=solid](0,0)(1,1)}
\rput{0}(108,19){\psellipse[](0,0)(7,-3)}
\rput{0}(104,19){\psellipse[linestyle=none,fillstyle=solid](0,0)(1,1)}
\rput{0}(108,19){\psellipse[linestyle=none,fillstyle=solid](0,0)(1,1)}
\rput{0}(112,19){\psellipse[linestyle=none,fillstyle=solid](0,0)(1,1)}
\rput(116.5,11){$X_n^+$}
\rput(116.5,16){$X_n^-$}
\rput(104,16){}
\psline[linewidth=0.4](104,16)(104,11)
\psline[linewidth=0.4](104,11)(112,16)
\psline[linewidth=0.4](112,16)(112,11)
\psline[linewidth=0.4](112,11)(104,16)
\rput(73.5,18.5){}
\rput{90}(83,13.5){\psellipse[](0,0)(13.5,-10.5)}
\rput{0}(81.5,8){\psellipse[](0,0)(7,-3)}
\rput{0}(77.5,8){\psellipse[linestyle=none,fillstyle=solid](0,0)(1,1)}
\rput{0}(81.5,8){\psellipse[linestyle=none,fillstyle=solid](0,0)(1,1)}
\rput{0}(85.5,8){\psellipse[linestyle=none,fillstyle=solid](0,0)(1,1)}
\rput{0}(81.5,19){\psellipse[](0,0)(7,-3)}
\rput{0}(77.5,19){\psellipse[linestyle=none,fillstyle=solid](0,0)(1,1)}
\rput{0}(81.5,19){\psellipse[linestyle=none,fillstyle=solid](0,0)(1,1)}
\rput{0}(85.5,19){\psellipse[linestyle=none,fillstyle=solid](0,0)(1,1)}
\rput(90,11){$X_1^+$}
\rput(90,16){$X_1^-$}
\rput(77.5,16){}
\psline[linewidth=0.4](77.5,16)(77.5,11)
\psline[linewidth=0.4](77.5,11)(85.5,16)
\psline[linewidth=0.4](85.5,16)(85.5,11)
\psline[linewidth=0.4](85.5,11)(77.5,16)
\rput(96.5,13){$\ldots$}
\rput(72,2){\Large{$X_1$}}
\rput(100,2){\Large{$X_n$}}
\rput(90.5,55){}
\rput(90.5,55){\Large{$q$}}
\rput(53,1){\Large{$p$}}
\end{pspicture}

%% file: clause.tex
\ifx\JPicScale\undefined\def\JPicScale{1}\fi
\psset{unit=\JPicScale mm}
\psset{linewidth=0.3,dotsep=1,hatchwidth=0.3,hatchsep=1.5,shadowsize=1,dimen=middle}
\psset{dotsize=0.7 2.5,dotscale=1 1,fillcolor=black}
\psset{arrowsize=1 2,arrowlength=1,arrowinset=0.25,tbarsize=0.7 5,bracketlength=0.15,rbracketlength=0.15}
\begin{pspicture}(0,0)(119.5,31)
\psline[linewidth=0.15](2,10)(85.5,10)
\psline[linewidth=0.15](2,15)(85.5,15)
\psline(2,20)(7,20)
\psline(84.5,20)(85.5,20)
\rput(4.5,21.5){$t_{1,1}$}
\rput(45.5,21.5){$t_{1,2}$}
\psline(28,25)(21,25)
\psline(21,26)(14,26)
\psline(14,25)(7,25)
\rput(10,27){$\alpha^1_{1,1}$}
\rput(24.5,27){$\alpha^1_{1,3}$}
\rput(17.5,28){$\alpha^{1}_{1,2}$}
\rput(0.5,10){$p$}
\rput(0.5,15){$q$}
\psline(28,3.5)(21,3.5)
\psline(21,4.5)(14,4.5)
\psline(14,3.5)(7,3.5)
\rput(10.5,5.5){$x^{+}_{1,1}$}
\rput(24.5,5.5){$x^{+}_{1,3}$}
\rput(17.5,6.5){$x^{+}_{1,2}$}
\psline(84.5,25)(77.5,25)
\psline(77.5,26)(70.5,26)
\psline(70.5,25)(63.5,25)
\psline(112.5,25)(105.5,25)
\psline(105.5,26)(98.5,26)
\psline(98.5,25)(91.5,25)
\psline(84.5,3.5)(77.5,3.5)
\psline(77.5,4.5)(70.5,4.5)
\psline(70.5,3.5)(63.5,3.5)
\psline(112.5,3.5)(105.5,3.5)
\psline(105.5,4.5)(98.5,4.5)
\psline(98.5,3.5)(91.5,3.5)
\rput(18,0.5){}
\rput(15.5,1){}
\psline[linewidth=0.15,linestyle=dotted](7,0)(7,31)
\psline[linewidth=0.15,linestyle=dotted](28,0)(28,31)
\psline[linewidth=0.15,linestyle=dotted](63.5,0)(63.5,31)
\psline[linewidth=0.15,linestyle=dotted](84.5,0)(84.5,31)
\psline[linewidth=0.15,linestyle=dotted](91.5,0)(91.5,31)
\psline[linewidth=0.15,linestyle=dotted](112.5,0)(112.5,31)
\psbezier[linewidth=0.15,fillcolor=white,fillstyle=solid](4.5,15)(4.5,15)(4.5,15)(4.5,15)
\psline(90.5,20)(91.5,20)
\psline[linestyle=dotted,fillcolor=white,fillstyle=solid](85.5,20)(90.5,20)
\psbezier[fillcolor=white,fillstyle=solid](28,25)(28,25)(28,25)(28,25)
\psline[fillcolor=white,fillstyle=solid](28,20)(63.5,20)
\psline(56,3.5)(49,3.5)
\psline(49,4.5)(42,4.5)
\psline(42,3.5)(35,3.5)
\rput(38.5,5.5){$x^{-}_{1,1}$}
\rput(52.5,5.5){$x^{-}_{1,3}$}
\rput(45.5,6.5){$x^{-}_{1,2}$}
\psline[linewidth=0.15,linestyle=dotted](85.5,15)(90.5,15)
\psline[linewidth=0.15,linestyle=dotted](85.5,10)(90.5,10)
\psline[linewidth=0.15](90.5,15)(113.5,15)
\psline[linewidth=0.15](90.5,10)(113.5,10)
\psline(112.5,20)(113.5,20)
\psline[linestyle=dotted,fillcolor=white,fillstyle=solid](114.5,20)(118.5,20)
\psline(118.5,20)(119.5,20)
\rput(88,21.5){$t_{1,3}$}
\rput(116,21.5){$t_{1,4}$}
\psline[linewidth=0.15,linestyle=dotted](113.5,15)(118.5,15)
\psline[linewidth=0.15,linestyle=dotted](113.5,10)(118.5,10)
\psline[linewidth=0.15](118.5,15)(119.5,15)
\psline[linewidth=0.15](118.5,10)(119.5,10)
\newrgbcolor{userLineColour}{0.6 0.6 0.6}
\newrgbcolor{userFillColour}{0.6 0.6 0.6}
\rput(88,5){$\ldots$}
\newrgbcolor{userLineColour}{0.6 0.6 0.6}
\newrgbcolor{userFillColour}{0.6 0.6 0.6}
\rput(88,4){}
\newrgbcolor{userLineColour}{0.6 0.6 0.6}
\newrgbcolor{userFillColour}{0.6 0.6 0.6}
\rput(116,5){$\ldots$}
\rput(67,5.5){$x^{+}_{2,1}$}
\rput(81,5.5){$x^{+}_{2,3}$}
\rput(74,6.5){$x^{+}_{2,2}$}
\rput(95,5.5){$x^{-}_{3,1}$}
\rput(109,5.5){$x^{-}_{3,3}$}
\rput(102,6.5){$x^{-}_{3,2}$}
\rput(66.5,27){$\alpha^2_{1,1}$}
\rput(81,27){$\alpha^2_{1,3}$}
\rput(74,28){$\alpha^{2}_{1,2}$}
\rput(94.5,27){$\alpha^3_{1,1}$}
\rput(109,27){$\alpha^3_{1,3}$}
\rput(102,28){$\alpha^{3}_{1,2}$}
\end{pspicture}

%% file: region_clause.tex
\ifx\JPicScale\undefined\def\JPicScale{1}\fi
\psset{unit=\JPicScale mm}
\psset{linewidth=0.3,dotsep=1,hatchwidth=0.3,hatchsep=1.5,shadowsize=1,dimen=middle}
\psset{dotsize=0.7 2.5,dotscale=1 1,fillcolor=black}
\psset{arrowsize=1 2,arrowlength=1,arrowinset=0.25,tbarsize=0.7 5,bracketlength=0.15,rbracketlength=0.15}
\begin{pspicture}(0,0)(119.5,31)
\newrgbcolor{userLineColour}{0.6 0.6 0.6}
\newrgbcolor{userFillColour}{0.6 0.6 0.6}
\pspolygon[linecolor=userLineColour,fillcolor=userFillColour,fillstyle=solid](84.75,31)(91.5,31)(91.5,20)(84.75,20)
\newrgbcolor{userLineColour}{0.8 0.8 0.8}
\newrgbcolor{userFillColour}{0.6 0.6 0.6}
\pspolygon[linewidth=0,linecolor=userLineColour,fillcolor=userFillColour,fillstyle=solid](2,20)(7,20)(7,31)(2,31)
\newrgbcolor{userLineColour}{0.6 0.6 0.6}
\newrgbcolor{userFillColour}{0.6 0.6 0.6}
\pspolygon[linecolor=userLineColour,fillcolor=userFillColour,fillstyle=solid](112.5,31)(119.5,31)(119.5,20)(112.5,20)
\newrgbcolor{userLineColour}{0.6 0.6 0.6}
\newrgbcolor{userFillColour}{0.6 0.6 0.6}
\pspolygon[linecolor=userLineColour,fillcolor=userFillColour,fillstyle=solid](98.5,31)(105.5,31)(105.5,26)(98.5,26)
\newrgbcolor{userLineColour}{0.6 0.6 0.6}
\newrgbcolor{userFillColour}{0.6 0.6 0.6}
\pspolygon[linecolor=userLineColour,fillcolor=userFillColour,fillstyle=solid](70.5,31)(77.5,31)(77.5,26)(70.5,26)
\newrgbcolor{userLineColour}{0.6 0.6 0.6}
\newrgbcolor{userFillColour}{0.6 0.6 0.6}
\pspolygon[linecolor=userLineColour,fillcolor=userFillColour,fillstyle=solid](105.5,31)(112.5,31)(112.5,25)(105.5,25)
\newrgbcolor{userLineColour}{0.6 0.6 0.6}
\newrgbcolor{userFillColour}{0.6 0.6 0.6}
\pspolygon[linecolor=userLineColour,fillcolor=userFillColour,fillstyle=solid](91.5,31)(98.5,31)(98.5,25)(91.5,25)
\newrgbcolor{userLineColour}{0.6 0.6 0.6}
\newrgbcolor{userFillColour}{0.6 0.6 0.6}
\pspolygon[linecolor=userLineColour,fillcolor=userFillColour,fillstyle=solid](77.5,31)(84.5,31)(84.5,25)(77.5,25)
\newrgbcolor{userLineColour}{0.6 0.6 0.6}
\newrgbcolor{userFillColour}{0.6 0.6 0.6}
\pspolygon[linecolor=userLineColour,fillcolor=userFillColour,fillstyle=solid](63.5,31)(70.5,31)(70.5,25)(63.5,25)
\newrgbcolor{userLineColour}{0.6 0.6 0.6}
\newrgbcolor{userFillColour}{0.6 0.6 0.6}
\pspolygon[linecolor=userLineColour,fillcolor=userFillColour,fillstyle=solid](28,31)(63.5,31)(63.5,20)(28,20)
\newrgbcolor{userLineColour}{0.6 0.6 0.6}
\newrgbcolor{userFillColour}{0.6 0.6 0.6}
\pspolygon[linecolor=userLineColour,fillcolor=userFillColour,fillstyle=solid](14,31)(28,31)(28,10)(14,10)
\psline[linewidth=0.15](2,10)(85.5,10)
\psline[linewidth=0.15](2,15)(85.5,15)
\psline(2,20)(7,20)
\psline(84.5,20)(85.5,20)
\rput(4.5,21.5){$t_{1,1}$}
\rput(45.5,21.5){$t_{1,2}$}
\psline(28,25)(21,25)
\psline(21,26)(14,26)
\newrgbcolor{userLineColour}{0.6 0.6 0.6}
\psline[linecolor=userLineColour](14,25)(7,25)
\rput(10,27){$\alpha^1_{1,1}$}
\rput(24.5,27){$\alpha^1_{1,3}$}
\rput(17.5,28){$\alpha^{1}_{1,2}$}
\rput(0.5,10){$p$}
\rput(0.5,15){$q$}
\psline(28,3.5)(21,3.5)
\psline(21,4.5)(14,4.5)
\psline(14,3.5)(7,3.5)
\rput(10.5,5.5){$x^{+}_{1,1}$}
\rput(24.5,5.5){$x^{+}_{1,3}$}
\rput(17.5,6.5){$x^{+}_{1,2}$}
\psline(84.5,25)(77.5,25)
\psline(77.5,26)(70.5,26)
\psline(70.5,25)(63.5,25)
\psline(112.5,25)(105.5,25)
\psline(105.5,26)(98.5,26)
\psline(98.5,25)(91.5,25)
\psline(84.5,3.5)(77.5,3.5)
\psline(77.5,4.5)(70.5,4.5)
\psline(70.5,3.5)(63.5,3.5)
\psline(112.5,3.5)(105.5,3.5)
\psline(105.5,4.5)(98.5,4.5)
\psline(98.5,3.5)(91.5,3.5)
\rput(18,0.5){}
\rput(15.5,1){}
\psline[linewidth=0.15,linestyle=dotted](7,0)(7,31)
\psline[linewidth=0.15,linestyle=dotted](28,0)(28,31)
\psline[linewidth=0.15,linestyle=dotted](63.5,0)(63.5,31)
\psline[linewidth=0.15,linestyle=dotted](84.5,0)(84.5,31)
\psline[linewidth=0.15,linestyle=dotted](91.5,0)(91.5,31)
\psline[linewidth=0.15,linestyle=dotted](112.5,0)(112.5,31)
\psbezier[linewidth=0.15,fillcolor=white,fillstyle=solid](4.5,15)(4.5,15)(4.5,15)(4.5,15)
\psline(90.5,20)(91.5,20)
\psline[linestyle=dotted,fillcolor=white,fillstyle=solid](85.5,20)(90.5,20)
\psbezier[fillcolor=white,fillstyle=solid](28,25)(28,25)(28,25)(28,25)
\psline[fillcolor=white,fillstyle=solid](28,20)(63.5,20)
\psline(56,3.5)(49,3.5)
\psline(49,4.5)(42,4.5)
\psline(42,3.5)(35,3.5)
\rput(38.5,5.5){$x^{-}_{1,1}$}
\rput(52.5,5.5){$x^{-}_{1,3}$}
\rput(45.5,6.5){$x^{-}_{1,2}$}
\psline[linewidth=0.15,linestyle=dotted](85.5,15)(90.5,15)
\psline[linewidth=0.15,linestyle=dotted](85.5,10)(90.5,10)
\psline[linewidth=0.15](90.5,15)(113.5,15)
\psline[linewidth=0.15](90.5,10)(113.5,10)
\psline(112.5,20)(113.5,20)
\psline[linestyle=dotted,fillcolor=white,fillstyle=solid](113.5,20)(118.5,20)
\psline(118.5,20)(119.5,20)
\rput(88,21.5){$t_{1,3}$}
\rput(116,21.5){$t_{1,4}$}
\psline[linewidth=0.15,linestyle=dotted](113.5,15)(118.5,15)
\psline[linewidth=0.15,linestyle=dotted](113.5,10)(118.5,10)
\psline[linewidth=0.15](118.5,15)(119.5,15)
\psline[linewidth=0.15](118.5,10)(119.5,10)
\newrgbcolor{userLineColour}{0.6 0.6 0.6}
\newrgbcolor{userFillColour}{0.6 0.6 0.6}
\rput(88,5){$\ldots$}
\newrgbcolor{userLineColour}{0.6 0.6 0.6}
\newrgbcolor{userFillColour}{0.6 0.6 0.6}
\rput(116.5,5){$\ldots$}
\rput(66.5,27){$\alpha^2_{1,1}$}
\rput(81,27){$\alpha^2_{1,3}$}
\rput(74,28){$\alpha^{2}_{1,2}$}
\rput(94.5,27){$\alpha^3_{1,1}$}
\rput(109,27){$\alpha^3_{1,3}$}
\rput(102,28){$\alpha^{3}_{1,2}$}
\rput(67,5.5){$x^{+}_{2,1}$}
\rput(81,5.5){$x^{+}_{2,3}$}
\rput(74,6.5){$x^{+}_{2,2}$}
\rput(95,5.5){$x^{-}_{3,1}$}
\rput(109,5.5){$x^{-}_{3,3}$}
\rput(102,6.5){$x^{-}_{3,2}$}
\end{pspicture}

%% file: variable.tex
\ifx\JPicScale\undefined\def\JPicScale{1}\fi
\psset{unit=\JPicScale mm}
\psset{linewidth=0.3,dotsep=1,hatchwidth=0.3,hatchsep=1.5,shadowsize=1,dimen=middle}
\psset{dotsize=0.7 2.5,dotscale=1 1,fillcolor=black}
\psset{arrowsize=1 2,arrowlength=1,arrowinset=0.25,tbarsize=0.7 5,bracketlength=0.15,rbracketlength=0.15}
\begin{pspicture}(0,0)(120,17.2)
\newrgbcolor{userLineColour}{0.6 0.6 0.6}
\newrgbcolor{userFillColour}{0.6 0.6 0.6}
\pspolygon[linecolor=userLineColour,fillcolor=userFillColour,fillstyle=solid](2,15)(5,15)(5,0)(2,0)
\newrgbcolor{userLineColour}{0.6 0.6 0.6}
\newrgbcolor{userFillColour}{0.6 0.6 0.6}
\pspolygon[linecolor=userLineColour,fillcolor=userFillColour,fillstyle=solid](38,15)(57.5,15)(57.5,0)(38,0)
\newrgbcolor{userLineColour}{0.6 0.6 0.6}
\newrgbcolor{userFillColour}{0.6 0.6 0.6}
\pspolygon[linecolor=userLineColour,fillcolor=userFillColour,fillstyle=solid](26,15)(31,15)(31,0)(26,0)
\newrgbcolor{userLineColour}{0.6 0.6 0.6}
\newrgbcolor{userFillColour}{0.6 0.6 0.6}
\pspolygon[linecolor=userLineColour,fillcolor=userFillColour,fillstyle=solid](19,3.5)(26,3.5)(26,0)(19,0)
\newrgbcolor{userLineColour}{0.6 0.6 0.6}
\newrgbcolor{userFillColour}{0.6 0.6 0.6}
\pspolygon[linecolor=userLineColour,fillcolor=userFillColour,fillstyle=solid](12,4.5)(19,4.5)(19,0)(12,0)
\newrgbcolor{userLineColour}{0.6 0.6 0.6}
\newrgbcolor{userFillColour}{0.6 0.6 0.6}
\pspolygon[linecolor=userLineColour,fillcolor=userFillColour,fillstyle=solid](5,3.5)(12,3.5)(12,0)(5,0)
\psline[linewidth=0.15](2,10)(52.5,10)
\psline[linewidth=0.15](2,15)(52.5,15)
\rput(0.5,10){$p$}
\rput(0.5,15){$q$}
\psline(26,3.5)(19,3.5)
\psline(19,4.5)(12,4.5)
\psline(12,3.5)(5,3.5)
\rput(8.5,5.5){$x^{+}_{1,1}$}
\rput(22.5,5.5){$x^{+}_{1,3}$}
\rput(15.5,6.5){$x^{+}_{1,2}$}
\rput(18,0.5){}
\rput(13.5,1){}
\psbezier[linewidth=0.15,fillcolor=white,fillstyle=solid](4.5,15)(4.5,15)(4.5,15)(4.5,15)
\psline(52,3.5)(45,3.5)
\psline(45,4.5)(38,4.5)
\newrgbcolor{userLineColour}{0.6 0.6 0.6}
\psline[linecolor=userLineColour](38,3.5)(31,3.5)
\psline[linewidth=0.15,linestyle=dotted](52.5,15)(56.5,15)
\psline[linewidth=0.15,linestyle=dotted](52.5,10)(56.5,10)
\psline[linewidth=0.15](56.5,15)(58,15)
\psline[linewidth=0.15](56.5,10)(57.5,10)
\psline[linewidth=0.15,linestyle=dotted,fillcolor=white,fillstyle=solid](5,0)(5,17)
\psline[linewidth=0.15,linestyle=dotted,fillcolor=white,fillstyle=solid](26,0)(26,17)
\psline[linewidth=0.15,linestyle=dotted,fillcolor=white,fillstyle=solid](31,0)(31,17)
\psline[linewidth=0.15,linestyle=dotted,fillcolor=white,fillstyle=solid](52,0)(52,17)
\rput(34.5,5.5){$x^{-}_{1,1}$}
\rput(48.5,5.5){$x^{-}_{1,3}$}
\rput(41.5,6.5){$x^{-}_{1,2}$}
\newrgbcolor{userLineColour}{0.6 0.6 0.6}
\newrgbcolor{userFillColour}{0.6 0.6 0.6}
\pspolygon[linecolor=userLineColour,fillcolor=userFillColour,fillstyle=solid](64.5,15)(67.5,15)(67.5,0)(64.5,0)
\newrgbcolor{userLineColour}{0.6 0.6 0.6}
\newrgbcolor{userFillColour}{0.6 0.6 0.6}
\pspolygon[linecolor=userLineColour,fillcolor=userFillColour,fillstyle=solid](100.5,15)(120,15)(120,0)(100.5,0)
\newrgbcolor{userLineColour}{0.6 0.6 0.6}
\newrgbcolor{userFillColour}{0.6 0.6 0.6}
\pspolygon[linecolor=userLineColour,fillcolor=userFillColour,fillstyle=solid](88.5,15)(93.5,15)(93.5,0)(88.5,0)
\newrgbcolor{userLineColour}{0.6 0.6 0.6}
\newrgbcolor{userFillColour}{0.6 0.6 0.6}
\pspolygon[linecolor=userLineColour,fillcolor=userFillColour,fillstyle=solid](81.5,15)(88.5,15)(88.5,3.5)(81.5,3.5)
\newrgbcolor{userLineColour}{0.6 0.6 0.6}
\newrgbcolor{userFillColour}{0.6 0.6 0.6}
\pspolygon[linecolor=userLineColour,fillcolor=userFillColour,fillstyle=solid](74.5,15)(81.5,15)(81.5,4.5)(74.5,4.5)
\newrgbcolor{userLineColour}{0.6 0.6 0.6}
\newrgbcolor{userFillColour}{0.6 0.6 0.6}
\pspolygon[linecolor=userLineColour,fillcolor=userFillColour,fillstyle=solid](67.5,15)(74.5,15)(74.5,3.5)(67.5,3.5)
\psline[linewidth=0.15](64.5,10)(115,10)
\psline[linewidth=0.15](64.5,15)(115,15)
\rput(63,10){$p$}
\rput(63,15){$q$}
\psline(88.5,3.5)(81.5,3.5)
\psline(81.5,4.5)(74.5,4.5)
\psline(74.5,3.5)(67.5,3.5)
\rput(71,5.5){$x^{+}_{1,1}$}
\rput(85,5.5){$x^{+}_{1,3}$}
\rput(78,6.5){$x^{+}_{1,2}$}
\rput(80.5,0.5){}
\rput(76,1){}
\psbezier[linewidth=0.15,fillcolor=white,fillstyle=solid](67,15)(67,15)(67,15)(67,15)
\psline(114.5,3.5)(107.5,3.5)
\psline(107.5,4.5)(100.5,4.5)
\newrgbcolor{userLineColour}{0.6 0.6 0.6}
\psline[linecolor=userLineColour](100.5,3.5)(93.5,3.5)
\psline[linewidth=0.15,linestyle=dotted](115,15)(119,15)
\psline[linewidth=0.15,linestyle=dotted](115,10)(119,10)
\psline[linewidth=0.15](119,15)(120,15)
\psline[linewidth=0.15](119,10)(120,10)
\psline[linewidth=0.15,linestyle=dotted,fillcolor=white,fillstyle=solid](67.4,-0.2)(67.4,17.2)
\rput(97,5.5){$x^{-}_{1,1}$}
\rput(111,5.5){$x^{-}_{1,3}$}
\rput(104,6.5){$x^{-}_{1,2}$}
\newrgbcolor{userLineColour}{0.6 0.6 0.6}
\newrgbcolor{userFillColour}{0.6 0.6 0.6}
\rput(117.5,4.5){$\ldots$}
\newrgbcolor{userLineColour}{0.6 0.6 0.6}
\newrgbcolor{userFillColour}{0.6 0.6 0.6}
\rput(54.5,4.5){$\ldots$}
\psline[linewidth=0.15,linestyle=dotted,fillcolor=white,fillstyle=solid](88.6,-0.4)(88.6,17)
\psline[linewidth=0.15,linestyle=dotted,fillcolor=white,fillstyle=solid](93.6,-0.2)(93.6,17.2)
\psline[linewidth=0.15,linestyle=dotted,fillcolor=white,fillstyle=solid](100.4,-0.2)(100.4,17.2)
\psline[linewidth=0.15,linestyle=dotted,fillcolor=white,fillstyle=solid](114.5,-0.4)(114.5,17)
\end{pspicture}

%% file: comp_and_perm_appendix.tex
\def\corel{\mathrel\Gamma}

\section{Proofs of Propositions \ref{prop:comp_ext} and \ref{prop:perm_ext}}

We describe in details how recognition algorithms for comparability and permutation graphs work, and show how to generalize them to obtain a partial orientation extension algorithm for comparability graphs and a partial representation extension algorithm for permutation graphs.
The former has been already sketched in Section \ref{sec:comparability}.

\heading{Recognition of \comp.} We define the relation $\corel$ on the edges. Let $x$, $y$ and $z$
be vertices. We write $xy \corel yz$ if and only if $xy \in E$, $yz \in E$ but $xz \notin E$. The
algorithm is based on the following key observation:

\begin{claim}
Let $x$, $y$ and $z$ be vertices of a comparability graph, such that $xy \corel yz$. Then every
transitive orientation of the graph orients the both $xy$ and $yz$ either to $y$, or from $y$.
\end{claim}

The algorithm works in the following way. We start with an undirected graph. We repeat the following step till
all the edges are oriented. In the beginning, the algorithm picks an arbitrary non-oriented edge and
orients it in an arbitrary direction. This maybe forces several other edges to be oriented in some
directions. According to transitivity and the relation $\corel$, we orient every edge for which the
direction is now given. If we are forced to change the direction of an edge, the algorithm fails.

Golumbic \cite{Golumbic1,Golumbic2} proved that this algorithm fails if and only if the graph is not a
comparability graph, independently on the choices of the edges in the first part of every step. A
straightforward implementation of this algorithm works in time $O((n+m)\Delta)$. We need to
orient every edge, which makes $O(n+m)$ steps. After orienting any edge we check all the incident
edges whether some other orientations are forced, which can be done in $O(\Delta)$. Also, this
algorithm allows to find other transitive orientations with a polynomial delay $O((n+m)\Delta)$,
by choosing different orientations of edges in the beginning of every step.

\heading{Permutation Graphs.} A permutation graph is given by a permutation $\pi$ of the elements
$\{1,\dots,n\}$. The vertices of the graph are the elements of the permutation. Two vertices $x$ and
$y$, $x < y$, are adjacent if and only if $\pi(x) > \pi(y)$. A pair of adjacent vertices is called
an \emph{inversion} of the permutation. We denote the class of permutation graphs by \perm.

We can represent a permutation graph as an intersection graph of segments in the plane, see
Figure~\ref{permutation_graph}. We place two copies of points $\{1,\dots,n\}$ on two parallel lines.
To a vertex $x$, we assign a segment from $x$ on the top line to $\pi(x)$ on the bottom line. It is
easy to show that permutations graphs are exactly intersection graphs of segments in the plane with
distinct ends touching two parallel lines.

\begin{figure}[t]
\centering
\includegraphics{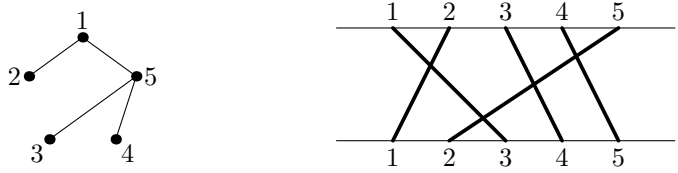}
\caption{A permutation graph of the permutation $\pi = (3,1,4,5,2)$ with an intersection
representation by segments touching two parallel lines.}
\label{permutation_graph}
\end{figure}

\heading{Recognition of \perm.} Their recognition is using comparability graphs, based on the
following characterization, by Even, Pnueli and Lempel~\cite{perm}: $$\perm = \comp
\cap \co\comp.$$ To recognize a permutation graph, it is sufficient to check whether both this graph
and its complement are comparability graphs. First, by constructing a segment representation, we
show that every graph from $\comp \cap \co\comp$ is also a permutation graph.

Let $G$ be a graph from $\comp \cap \co\comp$. We can transitively orient its edges
$\overrightarrow{E_1}$ and the edges of its complement $\overrightarrow{E_2}$. Together, the
orientations $\overrightarrow{E_1}$ and $\overrightarrow{E_2}$ form a transitive orientation of a
complete graph---a linear ordering. See Figure~\ref{perm_is_co_and_coco}.

\begin{figure}[b]
\centering
\includegraphics{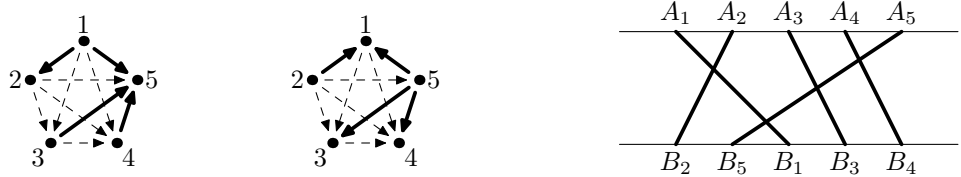}
\caption{Constructing a segment representation from orientations $\protect\overrightarrow{E_1}$ and
$\protect\overrightarrow{E_2}$.}
\label{perm_is_co_and_coco}
\end{figure}

Now, we have two parallel lines in the plane between which we want to place segments representing
the vertices of the graph. On the top line, we place points $A_1, \dots, A_n$ representing the
vertices in the linear ordering given by $\overrightarrow{E_1} \cup \overrightarrow{E_2}$. For the
bottom line, we reverse the orientation $\overrightarrow{E_1}$ and place points $B_1, \dots, B_n$
according to the linear ordering $\overleftarrow{E_1} \cup \overrightarrow{E_2}$. To a vertex $v$,
we assign the segment $A_vB_v$.

We need to verify that these segments give an intersection representation of the graph. If two
vertices are adjacent, their orders on the top line and on the bottom are different and the
corresponding segments intersect. On the other hand, if two vertices are non-adjacent, their orders
are the same and the segments do not intersect. So, we obtain a valid intersection representation of
the graph $G$ which implies $G$ is a permutation graph.

To obtain the other inclusion, notice that a segment representation gives transitive orientations
$\overrightarrow{E_1}$ and $\overrightarrow{E_2}$.  Another way is to observe that $\perm =
\co\perm$ (by reversing the bottom line) and $\perm \subseteq \comp$ (similarly to function graphs
described below).

\heading{Extending \perm.} A partial representation places several segments between two parallel
lines. The problem $\ext(\perm)$ asks whether it is possible to add the rest of the segments to
obtain a representation of a given permutation graph. Notice that all the segments have distinct
endpoints. 

Let $G$ be an input graph. This graph is a permutation graph if and only if $G$ and $\overline G$
are comparability graphs. The partial representation fixes directions of several edges of $G$ and
$\overline G$. We orient edges according to the ordering of the endpoints on the top line from left
to right. If two segments intersect, the corresponding edge is oriented in $G$, otherwise in
$\overline G$.

For the partially oriented $G$ and $\overline G$, we run the algorithm for comparability graph
extension which is described above. If extending is not possible, the algorithm fails.  Otherwise,
we obtain transitive orderings $\overrightarrow{E_1}$ and $\overrightarrow{E_2}$. We place the rest
of the points $A_1, \dots, A_n$ and $B_1, \dots, B_n$ in the correct order (i.e., $\overrightarrow{E_1} \cup \overrightarrow{E_2}$ on the top line and $\overleftarrow{E_1} \cup \overrightarrow{E_2}$ on the bottom one), and construct a
representation of $G$ as described above---a vertex $v$ is represented by a segment $A_vB_v$. 

\begin{proposition}
The described algorithm solves $\ext(\perm)$ correctly, and the running time is $O(n^3)$.
\end{proposition}

\begin{proof}
If $G$ has a segment representation extending the partial representation, then this representation
does not change the order of the placed segments.  Therefore, the corresponding transitive
orientations $\overrightarrow{E_1}$ and $\overrightarrow{E_2}$ have the edges oriented according to
the partial representation. It is possible to extend the partially oriented $G$ and $\overline G$.
Using these orientations, we construct a representation.

The running time is $O(n^3)$ since we need to run two instances of the algorithm solving
$\ext(\comp)$. To find all the representations, we use that the algorithm for $\ext(\comp)$ has the
same property.
\end{proof}